\documentclass[11pt, reqno]{amsart}
\usepackage[a4paper,hmargin=3cm,vmargin=3.5cm,heightrounded]{geometry} 

\usepackage{amsmath, amsthm, amscd, amsfonts, amssymb, graphicx, color, mathrsfs}
\usepackage{mathscinet}

\usepackage[utf8]{inputenc}
\usepackage{dsfont}
\usepackage{thmtools}
\usepackage{xcolor}
\usepackage{yfonts}
\usepackage{setspace}
\usepackage{chngcntr}
\usepackage{eurosym}
\usepackage{enumitem}
\usepackage[english]{babel}
\usepackage[colorlinks=true,linkcolor=blue]{hyperref}
\usepackage[all]{hypcap}
\usepackage[normalem]{ulem}

\newtheorem{theorem}{Theorem}[section]
\theoremstyle{definition}

\newtheorem{lemma}[theorem]{Lemma}
\newtheorem{example}[theorem]{Example}
\newtheorem{proposition}[theorem]{Proposition}

\newtheorem{remark}[theorem]{Remark}

\newcommand{\al}{\alpha}
\newcommand{\be}{\beta}
\newcommand{\de}{\delta}
\newcommand{\ep}{\varepsilon}
\newcommand{\ph}{\varphi}

\newcommand{\om}{\omega}

\newcommand{\si}{\sigma}

\newcommand{\la}{\lambda}

\newcommand{\Om}{\Omega}

\newcommand{\trace}{\operatorname{tr}}

\newcommand{\id}{\operatorname{id}}
\renewcommand{\d}{\text{\rm d}}

\newcommand{\Cpl}{\operatorname{Cpl}}
\newcommand{\mart}{\operatorname{Mart}}
\newcommand{\psd}{\operatorname{Sym}^d_+}
\newcommand{\sym}{\operatorname{Sym}^d}

\newcommand{\Cnt}{\operatorname{C}}

\newcommand{\Lip}{\operatorname{Lip}}

\newcommand{\cmp}{{\rm c}}
\newcommand{\bdd}{{\rm b}}

\newcommand{\Cb}{\Cnt_\bdd}
\newcommand{\Cc}{\Cnt_\cmp}
\newcommand{\Cbinf}{\Cnt_\bdd^\infty}
\newcommand{\Ccinf}{\Cnt_\cmp^\infty}
\newcommand{\Lipb}{\Lip_\bdd}

\newcommand{\E}{\mathbb{E}}
\newcommand{\N}{\mathbb{N}}
\renewcommand{\P}{\mathbb{P}}

\newcommand{\R}{\mathbb{R}}

\newcommand{\Bc}{\mathcal{B}}

\newcommand{\Pc}{\mathcal{P}}

\newcommand{\Wc}{\mathcal{W}}

\renewcommand{\epsilon}{\varepsilon}
\renewcommand{\rho}{\varrho}
\renewcommand{\phi}{\varphi}

\allowdisplaybreaks
\numberwithin{equation}{section}

\begin{document}

\title[Optimal transport and dynamic risk measures]{An optimal transport foundation for a class of dynamically consistent risk measures}

\author{Sven Fuhrmann}
\address{Department of Mathematics and Statistics, University of Konstanz}
\email{sven.fuhrmann@uni-konstanz.de}

\author{Michael Kupper}
\address{Department of Mathematics and Statistics, University of Konstanz}
\email{kupper@uni-konstanz.de}

\author{Max Nendel}
\address{Department of Statistics and Actuarial Science, University of Waterloo}
\email{mnendel@uwaterloo.ca}

\thanks{This paper is dedicated to Robert Denk on the occasion of his 60th birthday.\ We thank Robert for many inspiring discussions, his guidance, and his friendship.\ We thank Jonas Blessing for valuable comments and discussions related to this work.\ Financial support through the Deutsche Forschungsgemeinschaft (DFG, German Research Foundation) -- SFB 1283/2 2021 -- 317210226 is gratefully acknowledged.}

\date{\today}
\begin{abstract}
We study a class of dynamically consistent risk measures that robustify a time-homogeneous Markovian reference model by allowing for distributional uncertainty in its transition laws.\ We start from one-step convex risk evaluations in which ambiguity is captured by penalized worst-case expectations over alternative transition laws.\ Imposing time consistency then yields a convex monotone semigroup on bounded continuous payoff functions, and this semigroup represents the associated dynamic risk measure.\ The semigroup is uniquely characterized by its risk generator.\ Under a lower bound on the family of penalties in terms of suitable optimal transport costs relative to the reference laws, we identify the generator on smooth test functions.\ For optimal transport bounds with linear small-time scaling, this produces a first-order, drift-type correction given by a convex Hamiltonian acting on the gradient.\ Under martingale transport constraints and a different scaling, however, the leading correction is genuinely of second order and is described by a convex monotone functional acting on the Hessian.\ We illustrate both regimes for Wasserstein and martingale Wasserstein penalizations and derive explicit formulas via convex conjugates of the underlying transport costs.\ The associated dynamic risk measures admit stochastic control representations in which the control acts on the drift in the first-order case and on the volatility in the second-order case.\\[1mm]
\end{abstract}

\maketitle

\section{Introduction}
The quantification of risk is a central topic in financial mathematics.
Since probabilistic models are inherently subject to uncertainty, risk assessments based on a single reference probability measure may fail to capture imprecisions of the underlying model.
This has motivated the development of coherent and convex risk measures, cf.~\cite{arztner1999coherent,follmer2002convex,frittelli2002puttingorder}.
A common way to incorporate model uncertainty is to evaluate risk as a penalized worst-case expectation over probability measures in a neighbourhood of a reference model.

In this paper, we work with a Markovian reference model of the form $X_t^x:=\psi_t(x)+Y_t$, where $\psi_t(x)$ is a deterministic function of the current state $x\in\R^d$ and $Y_t$ has law $\mu_t$.
For a payoff $f(X_t^x)$ at time $t>0$, we consider the \emph{static} convex risk evaluation
\[
\rho^{\mathrm{stat}}\!\big(f(X_t^x)\big)
:= \sup_{\nu}\left(\int_{\mathbb{R}^d} f\big(\psi_t(x)+z\big)\,\nu(\d z)-\alpha_t(\nu)\right)
=:(I_t f)(x),
\]
where the supremum is taken over all probability measures on $\R^d$ and $\alpha_t$ penalizes deviations from the reference transition law $\mu_t$.\
Typical examples include divergence-based penalties, cf.\ \cite{FoellmerSchied2016,zbMATH06653590}, as well as transport-type penalties, cf.\ \cite{BartlEcksteinKupper2021,Foellmer2022,FuhrmannKupperNendel2023,KupperNenedelSgarabottolo2025,NendelSgarabottolo2022}.
These constructions are closely related to distributionally robust optimization, where ambiguity sets are modeled as transport balls and lead to tractable dual reformulations and sensitivity results, see, e.g.,~\cite{bartl2020robust,bartl2023sensitivity,BartlWiesel2023,blanchet2019quantifying,gao2016distributionally,esfahani2018data,nendel2025scaling}.

A central structural requirement in multi-period risk measurement is dynamic consistency.
In a time-homogeneous Markovian setting, this property takes the form of a tower principle:\ risk over a horizon $s+t$ is obtained by first evaluating the conditional risk over the remaining time interval of length $t$, given the information up to time $s$, and then applying the same risk evaluation once more.
Since the conditional distribution of the future increment depends on the past only through the state $X_s^x$, the conditional risk is again described by the one-step operator $I_t$, i.e., $\rho^{\mathrm{multi}}\!\big(f(X^x_{s+t})\big)=(I_s I_t f)(x)$.\ Iterating this argument along a partition of $[0,t]$ into $n$ subintervals of length $t/n$ and evaluating risk by backward recursion yields the $n$-step representation
\[
\rho^{\mathrm{multi}}\!\big(f(X_t^x)\big)=\big(I_{t/n}^n f\big)(x),
\]
showing that multi-period risk is given by the $n$-fold composition of the one-step operators.
This composition principle is consistent with the general theory of dynamic risk measures in discrete time, see, e.g.,~\cite{artzner2002coherentmultiperiod,bionnadal2006dynamic,cheridito2006dynamic,CheriditoKupper2011,follmer2006convexrm,riedel2004dynamic}.

The aim of this paper is to obtain an infinitesimal characterization of time-consistent
robust risk evaluations in continuous time. To this end, we start from the family of static one-step
operators $(I_t)_{t>0}$ and study the continuous-time limit of the time-consistent $n$-fold
compositions $I_{t/n}^n$ as $n\to\infty$.\ This limit connects the small-time behaviour of the robust one-step evaluations to a convex monotone
semigroup and leads to explicit formulas for  the infinitesimal generator of the semigroup. We thus define the associated dynamic risk measure by
\[
\rho^{\mathrm{dyn}}\!\big(f(X_t^x)\big)
:=\left(\lim_{n\to\infty} I_{t/n}^n f\right)(x)
=:\big(S_t f\big)(x),
\]
whenever the limit exists.\ In the linear case, such limits are classical and covered by Chernoff’s product formula, cf.\ \cite{Chernoff68,Chernoff74}.
Here, $(I_t)_{t>0}$ is convex and monotone on $C_b(\R^d)$. The limit family $(S_t)_{t\ge0}$ forms a {convex monotone semigroup}, i.e., a nonlinear analogue of a (sub-)Markovian semigroup.\ Such semigroups frequently appear in the context of stochastic optimal control and Hamilton-Jacobi-Bellman equations, cf.\ \cite[Section II.3]{zbMATH05032360}.\
 A key subtlety compared to the linear theory is that the domain of the infinitesimal generator is typically not invariant under the related convex monotone semigroup.\
Recent work has therefore developed a semigroup theory tailored to convex monotone operators, cf.~\cite{BlessingDenkKupperNendel2025,BlessingKupperNendel2025}.\ It relies on compactness in the mixed topology and comparison principles based on the notion of a $\Gamma$-generator, which is an extension of the generator to an appropriate Lipschitz set, which in the linear case is commonly known as the Favard space.\ Within this framework, Chernoff-type approximation theorems provide verifiable conditions for the convergence of discretization schemes and yield an infinitesimal characterization of the limiting dynamics, cf.~\cite{BlessingDenkKupperNendel2025,BlessingKupper2023,BlessingKupper2025,BlessingKupperNendel2025}.\ In particular, \cite{BlessingDenkKupperNendel2025} implies that $(S_t)_{t\ge0}$ is uniquely determined by its generator on smooth test functions, and the generator agrees with the right-derivative at zero of the one-step evaluations, i.e.,
\[
Af=\lim_{h\downarrow0}\frac{S_h f-f}{h}=\lim_{h\downarrow0}\frac{I_h f-f}{h}=:I'(0)f,
\]
whenever the limit on the right-hand side exists.\
Hence, the associated dynamic risk measure is fully determined by the {risk generator} $I'(0)$.

\subsection{Main results.}
We establish minimal conditions on the  family of penalizations $(\alpha_t)_{t>0}$ under which the risk generator exists on smooth test functions and admits an explicit representation.\ This yields an infinitesimal characterization of a broad class of time-homogeneous dynamically consistent risk measures whose penalties are bounded from below by optimal transport costs.
 Concretely, we assume that for each $t>0$ and each probability measure $\nu$,
\[
\alpha_t(\nu)\ge \inf_{\pi\in\Cpl(\mu_t,\nu)}\int_{\R^d\times \R^d} t\, c\bigg(\frac{|z-y|}{t}\bigg)\,\pi(\d y,\d z),
\]
for some nondecreasing function $c\colon[0,\infty)\to\R$ satisfying $\lim_{v\to\infty}c(v)/v=\infty$. Our first main result, Theorem~\ref{thm.generator}, shows that, under this linear transport scaling, the limiting generator is a first-order perturbation
\[
Af = Lf + g(\nabla f),
\]
where $L$ is the generator of the Markovian reference dynamics $\psi_t(x)+Y_t$ and $g$ is a convex function capturing the maximal penalized small-time directional displacement of the reference law.\ The existence and explicit identification of $L$ build on \cite{nendel2025chernoff}, which develops a Chernoff--Mehler product approximation for models of the form $X_t^x=\psi_t(x)+Y_t$. Under verifiable regularity and tightness assumptions, it establishes convergence of the associated products to the transition semigroup of a L\'evy process with drift and identifies the limiting generator from the small-time behaviour of the approximating laws. In particular, the framework accommodates deterministic components $\psi_t$ given by flows of Lipschitz ODEs, as well as Euler and Runge--Kutta approximations thereof.\
The nonlinear correction is described through the functions
\[
g_t(m)
:=\sup_{\nu}
\bigg(\int_{\R^d\times\R^d}\langle m,z-y\rangle\,\pi_{t,\nu}^*(\d y,\d z)-\alpha_t(\nu)\bigg),
\]
where the supremum is taken over all probability measures $\nu$ on $\R^d$ and $\pi_{t,\nu}^*$ is a suitably chosen coupling of $\mu_t$ and $\nu$, which
are shown to satisfy $g_t(m)/t\to g(m)$ as $t\downarrow0$, where $g$ is convex and controlled by the conjugate cost $c^\ast$. Heuristically, this reflects that the penalty acts on deviations of the one-step reference law $\mu_t$, so the leading nonlinear contribution is driven by admissible displacements $z-y$ of stochastic increments, while the state dynamics $\psi_t(x)$ enters through the linear part $L$.

Our second main result, Theorem~\ref{thm.generator2}, treats martingale transport penalties.\ Here, deviations from $\mu_t$ are again controlled from below by transport costs, but with a different small-time scaling under a martingale constraint. In this regime, the leading correction corresponds to modifications of the volatility rather than the drift, and the risk generator becomes a second-order perturbation
\[
Af \;=\; Lf + G(\nabla^2 f),
\]
where $G$ is a convex monotone functional on the vector space of symmetric matrices, endowed with the Loewner order.\ The functional $G$ quantifies the maximal penalized modification of local covariance compatible with the martingale constraint and, in many examples, admits an explicit representation in terms of the convex conjugate of the underlying cost. We illustrate both regimes through Wasserstein-type and martingale Wasserstein-type penalizations.\ This links transport-based local model uncertainty to Hamilton--Jacobi-type equations in the first-order case and to fully nonlinear parabolic (integro)-differential equations in the second-order case.

The final part of the paper passes from these generator formulae to the continuous-time dynamic risk measures. Section~\ref{sec.chernoff} identifies the Chernoff limit of the iterated evaluations $I_{t/n}^n$ in terms of stochastic control problems. In the first-order case, the resulting semigroup has the representation
\[
(S_tf)(x)
=
\sup_{\beta}
\E\bigg[
f\big(X_t^{x,\beta}\big)
-\int_0^t g^*(\beta_s)\,\d s
\bigg],
\]
where $g^*$ is the convex conjugate of $g$ and $X^{x,\beta}$ denotes the controlled reference dynamics with additional drift control $\beta$.\ We point out that, in the Brownian case, this is the Markovian control representation associated with convex $g$-expectations.

In the second-order martingale case, the martingale constraint removes first-order displacements, and the corresponding semigroup is represented by
\[
(S_tf)(x)
=
\sup_\si
\E\bigg[
f\big(X_t^{x,\si}\big)
-\int_0^t G^*(\si_s)\,\d s
\bigg],
\]
where $G^*$ is the convex conjugate of $G$ with respect to the trace pairing and $X^{x,\si}$ denotes the controlled reference dynamics with additive local covariance control $\si$.\ In the Brownian case, this yields a control representation, which is akin to the $G$-expectation.

\subsection{Related literature}
A prominent approach to {dynamic} risk measurement is based on backward stochastic differential equations (BSDEs) in a Brownian filtration. Beginning with $g$-expectations, BSDEs provide time-consistent nonlinear evaluations and, for convex drivers, yield convex dynamic risk measures; cf., e.g., \cite{coquet2022filtration,elkaroui1997bsde,peng2004nonlinearexpectations,rosazzagianin2006riskmeasures,xu2016multidimensional}. In this setting, the local structure of the dynamic evaluation is encoded directly in the BSDE driver. Moreover, representation results connect drivers with dynamic penalties and dual formulations under suitable domination or absolute continuity assumptions, cf.\ \cite{delbaen2010representation} and the references therein.

To address {model uncertainty} in continuous time, a related literature develops second-order BSDEs (2BSDEs) and {nonlinear expectations}.\ These frameworks are closely tied to volatility uncertainty and so-called $G$-expectations, and naturally lead to fully nonlinear second-order PDEs, see, e.g., \cite{cheridito2007secondorder,kazitani2015secondorderbsde,nutz2012superhedgingdynamic,peng2007Gexpectation,soner2012wellposedness}.\ Section~\ref{sec.chernoff} makes this link explicit. Ordinary transport scaling leads to a first-order Hamiltonian, which  corresponds to $g$-expectations, while martingale transport scaling leads to a second-order Hamiltonian, which corresponds to $G$-expectations. Dynamic consistency has also been studied from an axiomatic and stability viewpoint.\ In this direction, dynamic risk measures are characterized through stability properties of their representing sets, most notably $m$-stability and related pasting or concatenation conditions, see, e.g., \cite{acciaio2011dynamicrm,bionnadal2009timeconsistent, cheridito2006dynamic,CheriditoKupper2011,delbaen2006mstable,follmer2006convexrm,riedel2004dynamic}.\ Such results clarify when recursive evaluation is possible and connect dynamic risk measurement to robust dynamic programming principles.

Closer in spirit to the present paper are contributions that seek an {infinitesimal} description of dynamic risk evaluation via generator-type objects.\ For instance, the work \cite{pichler2022quantification} derives a risk generator as the small-time limit of iterated coherent risk evaluations and identifies it for It\^o diffusions via the local characteristics of the dynamics.\ In the diffusive setting, considered therein, this yields a semilinear ``risk-adjusted'' generator, i.e., the classical diffusion generator plus a correction term given by the absolute value of a volatility-weighted gradient term.\ A similar viewpoint also appears in the Brownian BSDE or $g$-expectation framework.\ The paper \cite{delbaen2010representation} connects the local driver to the dynamic penalty (and dual objects) through representation results. In the case of Markov semigroups consisting of sublinear operators, \cite{kuhn2018viscosity} and \cite{hu2009glev} derive an explicit characterization of the sublinear generator as a supremum of linear generators. In contrast to these works, we consider a Markovian setting that is not tied to a Brownian filtration and allows for drift and jump components, and we focus on robustness via penalties bounded from below by optimal transport costs.\ Our main results show that these penalties lead again to local corrections in terms of first-order Hamiltonians (drift uncertainty), for classical optimal transport costs, and second-order Hessian terms (volatility uncertainty) in the case of martingale optimal transport.

Finally, recursive utility and stochastic differential utility provide another influential paradigm for dynamic evaluations in economic theory and finance.\ In continuous time, intertemporal preferences are specified through a local aggregator, leading to time-consistent nonlinear valuations that admit BSDE-type characterizations in Brownian settings, cf.~\cite{DuffieEpstein1992,Lazrak2004,schroder1999optimal}.\ While the economic interpretation differs from risk measurement, both approaches share that the evaluation is determined by an infinitesimal object. The present paper contributes by deriving such a local characterization from transport-controlled robustness in Markovian models, thus linking transport-based ambiguity to explicit generator corrections.

The remainder of the paper is organized as follows.\
Section~\ref{sec.setup} introduces the Markovian reference model, the one-step risk operators, and the standing assumptions.
In Section~\ref{sec:first order}, we study first-order transport-penalized perturbations and identify the corresponding drift-type correction in the risk generator.
Section~\ref{sec:second order} is devoted to martingale optimal transport penalizations, where the leading correction is a nonlinear second-order term acting on the Hessian.\ Section~\ref{sec.chernoff} proves the nonlinear Chernoff approximation for the iterated one-step evaluations and identifies the limiting semigroups as stochastic control problems for drift and volatility uncertainty.\ The appendix contains an elementary auxiliary result for convex functions on the real line
as well as localization estimates for convex monotone semigroups used in the Chernoff approximation and the control representations.

\section{Setup and preliminaries}\label{sec.setup}
Throughout, let $d\in\N$ and equip $\R^d$ with the Euclidean norm $|\cdot|$, the
standard inner product $\langle\cdot,\cdot\rangle$, and the Borel $\sigma$-algebra
$\Bc(\R^d)$.
We write $\Pc(\R^d)$ for the set of all Borel probability measures on $\R^d$ and
$\Cb(\R^d)$ for the space of all bounded continuous functions $f\colon\R^d\to\R$.
For a bounded map $f\colon\R^d\to\R^m$ with $m\in\N$, we set
$\|f\|_\infty:=\sup_{x\in\R^d}|f(x)|$.

We work on $\Cb(\R^d)$ endowed with the {mixed topology}, also called the
{strict topology}.  For our purposes it suffices to work
with the associated notion of sequential convergence. Given
$(f_n)_{n\in\N}\subset\Cb(\R^d)$ and $f\in\Cb(\R^d)$, we write $f_n\to f$ if
\begin{equation}\label{eq.conv.mixed}
\sup_{n\in\N}\|f_n\|_\infty<\infty
\quad\text{and}\quad
\lim_{n\to\infty}\sup_{|x|\le r}\big|f_n(x)-f(x)\big|=0
\quad\text{for all }r\ge0.
\end{equation}
The mixed topology is convenient for the analysis of
convex monotone semigroups, since the relevant operators are typically
{sequentially continuous} with respect to \eqref{eq.conv.mixed} and the
required compactness arguments are naturally formulated in this setting, see,
e.g., \cite{BlessingDenkKupperNendel2025,BlessingKupperNendel2025}.
For a family $(f_t)_{t>0}\subset\Cb(\R^d)$ and
$f\in\Cb(\R^d)$, we write $f_t\to f$ as $t\downarrow 0$ if
$f_{t_n}\to f$ for every sequence $(t_n)_{n\in\N}\subset(0,\infty)$ with
$t_n\to 0$.

We denote by $\Lip=\Lip(\R^d)$ the space of all Lipschitz continuous functions
$f\colon\R^d\to\R$, meaning that there exists $L\ge0$ such that
\begin{equation}\label{eq.lip}
|f(x_1)-f(x_2)|\le L|x_1-x_2|,\qquad x_1,x_2\in\R^d.
\end{equation}
For a Lipschitz map $f\colon\R^d\to\R^m$ with $m\in\N$, we define its (optimal)
Lipschitz constant by
\[
\|f\|_{\Lip}
:=\inf\Big\{L\ge0 \,\Big|\, |f(x_1)-f(x_2)|\le L|x_1-x_2|
\text{ for all }x_1,x_2\in\R^d\Big\}.
\]
Moreover, $\Lipb$ denotes the space of all bounded Lipschitz
functions $\R^d\to \R$, that is, the space of all $f\in\Lip$ with $\|f\|_\infty<\infty$.

For $\mu,\nu\in\Pc(\R^d)$, we denote by $\Cpl(\mu,\nu)$ the set of all couplings of
$\mu$ and $\nu$, i.e., all $\pi\in\Pc(\R^d\times\R^d)$ whose first marginal equals
$\mu$ and whose second marginal equals $\nu$.
We further denote by $\mart(\mu,\nu)$ the set of all martingale couplings between
$\mu$ and $\nu$, i.e., all $\pi\in\Cpl(\mu,\nu)$ such that, for every
$\varphi\in\Cb(\R^d;\R^d)$ for which the integral below is well-defined,\footnote{That is,
$\int_{\R^d\times\R^d}|\langle\varphi(y),z-y\rangle|\,\pi(dy,dz)<\infty$.}
one has
\begin{equation}\label{eq.def.mart.coupling}
\int_{\R^d\times\R^d}\langle \varphi(y), z-y\rangle\,\pi(\d y,\d z)=0.
\end{equation}
Throughout, let $(\mu_t)_{t>0}\subset \Pc(\R^{d})$ be a family of probability measures on $\Bc(\R^d)$ and let $(\psi_t)_{t>0}$ be a family of maps $\psi_t:\R^d\to\R^d$. We assume:

\begin{enumerate}
    \item[(A)]\makeatletter\def\@currentlabel{A}\makeatother\label{cond.A}
    It holds that $\mu_h\to \delta_0$ as $h\downarrow 0$ and, for every $R\ge 0$,
    \begin{equation}\label{eq.ass.psi_h}
        \sup_{|x|\le R}\,|\psi_h(x)-x|\;\longrightarrow\;0
        \quad\text{as }h\downarrow 0.
    \end{equation}
\end{enumerate}

For $t>0$ and $x\in\R^d$, let $Y_t$ be an $\R^d$-valued random variable with law $\mu_t$, and define
\[
    X_t^x:=\psi_t(x)+Y_t .
\]
We then introduce the family $(P_t)_{t> 0}$ of operators on $\Cb(\R^d)$ by
\begin{equation}\label{eq:reference}
    (P_t f)(x):=\E\big[f(X_t^x)\big]
    =\int_{\R^d} f\big(\psi_t(x)+y\big)\,\mu_t(\d y),
\end{equation}
for $t>0$, $f\in\Cb(\R^d)$, and $x\in\R^d$.\ We point out that the family $(P_t)_{t> 0}$ is {not} necessarily the transition semigroup of a Markov process.\ However, in \cite{nendel2025chernoff}, it is shown that, under suitable stronger conditions than Condition \eqref{cond.A}, the family $(P_t)_{t>0}$ is,  up to a subsequence, Chernoff equivalent to a L\'evy process with drift.\ We refer to \cite{nendel2025chernoff} for the details.

    \section{First-order perturbation}\label{sec:first order}    
In this section, we analyze the first-order behavior of a nonlinear
perturbation of the reference semigroup \eqref{eq:reference}.
To this end, we introduce a convex analogue of the operators
\eqref{eq:reference} via a family $(I_t)_{t\ge 0}$ of operators
on $\Cb(\R^d)$.\
For $t>0$, the operator $I_t$ is defined by		\begin{equation}\label{eq.operatorI}
		(I_tf)(x) := \sup_{\nu \in \Pc(\R^d)} \left(\int_{\R^d} f(\psi_{t}(x)+z)\nu(\d z) - \alpha_{t}(\nu)\right)\quad\text{for all }f\in \Cb(\R^d)\text{ and }x\in \R^d.
	\end{equation}
Here, $(\alpha_t)_{t> 0}$ is a family of penalizations $\Pc(\R^d) \to [0,\infty] $ subject to the following assumption: 
\begin{enumerate}
    \item[(P)]\makeatletter\def\@currentlabel{P}\makeatother\label{cond.P} There exist $h_0>0$ and a nondecreasing function $c\colon [0,\infty)\to \R$ with \[
\lim_{v\to\infty}\frac{c(v)}{v}=\infty
\] and
    \[
        \alpha_h(\nu)\geq \inf_{\pi\in \Cpl(\mu_h,\nu)}\int_{\R^d\times \R^d} h c\bigg(\frac{|z-y|}{h}\bigg)\,\pi(\d y,\d z)\quad\text{for all }h\in (0,h_0)\text{ and }\nu\in \Pc(\R^d)
    \]
    Moreover, $\alpha_h(\mu_h)=0$ for all $h\in (0,h_0)$.
\end{enumerate}	

 Replacing $c$ by $c^{**}\leq c$, see Lemma \ref{lem.conjugate} b), we may w.l.o.g.\ assume that $c$ is convex, so that $c$ is continuous and $(0,\infty)\to \R,\; v\mapsto \frac{c(v)}{v}$ is nondecreasing since, by Assumption \eqref{cond.P}, $c(0)\leq 0$.\footnote{Since $c$ is nondecreasing,
     $c(0)\leq \inf_{\pi\in \Cpl(\mu_h,\mu_h)}\int_{\R^d\times \R^d} h c\big(\frac{|z-y|}{h}\big)\,\pi(\d y,\d z)\leq \alpha_h(\mu_h)=0$ for $h\in (0,h_0)$.}\ 
     We define the convex conjugate of $c$ by
\[
c^\ast(w):=\sup_{v\ge 0}\big(vw-c(v)\big) \quad\text{for all } w\ge 0.
\]
By Lemma \ref{lem.conjugate} a), the supremum is finite for every $w\ge 0$, so that $c^\ast(w)\in [0,\infty)$ for all $w\geq 0$.

    Observe that probability measures $\nu\in \Pc(\R^d)$ with $\alpha_t(\nu)=\infty$ do not affect the value
of the supremum in \eqref{eq.operatorI}, so that the supremum may equivalently
be taken over the effective domain
\[
\Pc_t := \big\{\nu\in\Pc(\R^d)\,\big|\,\alpha_t(\nu)<\infty\big\}\quad \text{for all }t>0.
\]
This convention will be used throughout this section.\ We start with a series of observations.
\begin{remark}\label{rem.ass.firstorder}\
Assume that Condition \eqref{cond.P} is satisfied.
\begin{enumerate}
    \item[a)] Since $c$ is continuous and nondecreasing, by \cite[Theorem 4.1]{villani2008optimal}, for all $h\in (0,h_0)$ and $\nu\in \Pc_h$, there exists an optimal coupling $\pi^*_{h,\nu}\in \Cpl(\mu_h,\nu)$ with 
    \begin{align}
        \notag \int_{\R^d\times \R^d} hc\bigg(\frac{|z-y|}{h}\bigg)\,\pi^*_{h,\nu}(\d y,\d z)&=\inf_{\pi\in \Cpl(\mu_h,\nu)}\int_{\R^d\times \R^d} hc\bigg(\frac{|z-y|}{h}\bigg)\,\pi(\d y,\d z)\\
        &\leq \alpha_h(\nu)<\infty. \label{eq.opt.coupling}
        \end{align}   
    \item[b)] Let $f \in \Lip(\mathbb{R}^d)$ and $h\in (0,h_0)$.\ Then, the assumptions on the family $(\mu_t)_{t>0}$ do, in general, not ensure that $f$ is $\mu_t$-integrable for any $t>0$.\ However, for all $\nu \in \mathcal{P}_h$ and any coupling $\pi \in \Cpl(\mu_h,\nu)$ with
       \[
         \int_{\R^d\times \R^d} hc\bigg(\frac{|z-y|}{h}\bigg)\,\pi(\d y,\d z)<\infty,
        \]
        it holds
        \begin{align}
          \notag \int_{\R^d\times \R^d} |f(z)-f(y)|\,\pi(\d y,\d z)&\leq h\int_{\R^d\times \R^d} \|f\|_{\Lip}\frac{|z-y|}h-c\bigg(\frac{|z-y|}h\bigg)\,\pi(\d y,\d z)\\
          \notag &\quad + \int_{\R^d\times \R^d} hc\bigg(\frac{|z-y|}{h}\bigg)\,\pi(\d y,\d z)\\
          &\leq h c^*(\|f\|_{\Lip})+\int_{\R^d\times \R^d} hc\bigg(\frac{|z-y|}{h}\bigg)\,\pi(\d y,\d z)<\infty, \label{eq.wass.finite}
        \end{align}
        so that the mapping $\R^d\times \R^d \to \R,\; (y,z)\mapsto f(z)-f(y)$ is $\pi$-integrable.\ Moreover, by \eqref{eq.wass.finite}, for all $\nu\in \Pc_h$ and any coupling $\pi\in  \Cpl(\mu_h,\nu)$ with 
        \begin{equation}\label{eq.almost.opt.coupling}
   \int_{\R^d\times \R^d} hc\bigg(\frac{|z-y|}{h}\bigg)\, \pi(\d y,\d z)\leq \al_h(\nu),
\end{equation}
       it holds
\begin{equation}\label{eq.apriori.lipschitz}
         \int_{\R^d\times \R^d} |f(z)-f(y)|\,\pi(\d y, \d z )-\alpha_h(\nu)\leq hc^*\big(\|f\|_{\Lip}\big).
          \end{equation}
          In particular, if $f$ is also bounded, by \eqref{eq.opt.coupling},
          \begin{align}
  \notag        0&\leq (I_h f)(x)-(P_h f)(x)\\
\notag &\qquad \quad =
\sup_{\nu\in\Pc_h}
\bigg(
\inf_{\pi\in\Cpl(\mu_h,\nu)}
\int_{\R^d\times \R^d} f\big(\psi_h(x)+z\big)-f\big(\psi_h(x)+y\big)\,\pi(\d y,\d z)
-\alpha_h(\nu)
\bigg) \label{eq.apriori.lipschitz.thm}\\
&\qquad\quad \le hc^*\big(\|f\|_{\Lip}\big) \quad \text{for all }x\in \R^d.
\end{align}
\item[c)] Let $f \in \Lip(\mathbb{R}^d)$, $h\in (0,h_0)$, $\nu\in \Pc_h$.
Then, for all $\pi\in\Cpl(\mu_h,\nu)$ with
\begin{equation}\label{eq.finite.first.moment}
\int_{\R^d\times\R^d} |z-y|\,\pi(\d y,\d z)<\infty,
\end{equation}
the integral
\[
\int_{\R^d\times\R^d} f(y)-f(z)\,\pi(\d y,\d z)
\]
is well-defined and independent of the choice of $\pi\in\Cpl(\mu_h,\nu)$ among all couplings satisfying the integrability condition \eqref{eq.finite.first.moment}.\ Indeed, for $\pi\in\Cpl(\mu_h,\nu)$ with \eqref{eq.finite.first.moment},
\[
\int_{\R^d\times\R^d}|f(z)-f(y)|\,\pi(\d y,\d z)\leq \|f\|_{\Lip} \int_{\R^d\times\R^d} |z-y|\,\pi(\d y,\d z)<\infty.
\]
Now, let $\pi_1,\pi_2\in\Cpl(\mu_h,\nu)$ satisfy \eqref{eq.finite.first.moment} and define the truncation
\[
f_n(x):=(-n)\vee f(x)\wedge n\quad \text{for all }n\in \N\text{ and }x\in \R^d.
\]
Then, for all $n\in \N$, $f_n$ is bounded and continuous, so that
\begin{align*}
 \int_{\R^d\times\R^d}f_n(z)-f_n(y)\,\pi_1(\d y,\d z)&= \int_{\R^d}f_n(z)\,\nu(\d z)-\int_{\R^d}f_n(y)\,\mu_h(\d y)\\
 &=\int_{\R^d\times\R^d}f_n(z)-f_n(y)\,\pi_2(\d y,\d z).
\end{align*}
Since the truncation is $1$-Lipschitz, we have
$|f_n(z)-f_n(y)|\le |f(z)-f(y)|$ for all $y,z\in \R^d$.\ Hence, by dominated convergence,
\begin{align*}
 \int_{\R^d\times\R^d}f(z)-f(y)\,\pi_1(\d y,\d z)&=\lim_{n\to\infty}\int_{\R^d\times\R^d}f_n(z)-f_n(y)\,\pi_1(\d y,\d z)\\
 &=\lim_{n\to\infty}\int_{\R^d\times\R^d}f_n(z)-f_n(y)\,\pi_2(\d y,\d z)\\
 &= \int_{\R^d\times\R^d}f(z)-f(y)\,\pi_2(\d y,\d z).
\end{align*}
 \end{enumerate}
 \end{remark}

We now identify the first-order limit of the rescaled nonlinear perturbations.
For each $h\in (0,h_0)$, we introduce an auxiliary function $g_h$ which captures the
maximal first-order displacement induced by perturbations of the
reference measure $\mu_h$, penalized by the transport cost $\alpha_h$.

To that end, for all $h\in (0,h_0)$ and all $\nu\in \Pc_h$, we fix an arbitrary coupling $\pi_{h,\nu}^*\in \Cpl(\mu_h,\nu)$ with \eqref{eq.finite.first.moment} and consider the map $g_h\colon \R^d\to \R$, given by 	\begin{equation}\label{eq.def.g_h}
		g_h(m) :=  \sup_{\nu \in \Pc_h}\bigg(\int_{\R^{d}\times \R^{d}} \langle m,z-y\rangle\,\pi^*_{h,\nu}(\d y,\d z) - \alpha_h(\nu)\bigg)\quad \text{for all }m\in \R^d.
	\end{equation}
By Remark \ref{rem.ass.firstorder} c), the definition of the function $g_h$ is independent of the choice of the coupling $\pi_{h,\nu}^*\in \Cpl(\mu_h,\nu)$ among all couplings satisfying the integrability condition \eqref{eq.finite.first.moment} for all $h\in (0,h_0)$ and $\nu\in \Pc_h$.\ In particular, we may assume w.l.o.g.\ that 
\[
   \int_{\R^d\times \R^d} hc\bigg(\frac{|z-y|}{h}\bigg)\, \pi_{h,\nu}^*(\d y,\d z)\leq \al_h(\nu),
\]
see Remark \ref{rem.ass.firstorder} a), i.e., $\pi_{h,\nu}^*$ satisfies \eqref{eq.almost.opt.coupling} in Remark \ref{rem.ass.firstorder} b) for all $h\in (0,h_0)$ and $\nu\in \Pc_h$.

Let $h\in (0,h_0)$ and $m\in \R^d$.\ Then, the quantity $g_h(m)$ measures the maximal average
displacement in direction $m$ that can be produced by transporting the reference
measure $\mu_h$ to probability measures $\nu\in \Pc_h$ at scale $h$, subject to the cost
constraint $\alpha_h$.\ For a given probability measure $\nu\in \Pc_h$, the coupling
$\pi_{h,\nu}^*\in\Cpl(\mu_h,\nu)$, satisfying \eqref{eq.almost.opt.coupling}, represents an efficient transport plan and the
integral $\int_{\R^d\times\R^d}\langle m,z-y\rangle\,\pi_{h,\nu}^*(\d y,\d z)$
corresponds to the resulting mean directional drift.

The following auxiliary result provides a crucial a priori estimate. 

\begin{lemma}\label{lem.apriori1}
 Assume that Condition \eqref{cond.P} is satisfied.\ Then, for all $L\geq 0$, there exists $a\geq 0$ such that, for all $h\in (0,h_0)$ and $f\in \Lip(\R^d)$ with $\|f\|_{\Lip}\leq L$,
 \begin{align}
  \notag\sup_{\nu\in \Pc_h}\bigg(\int_{\R^d\times \R^d} &f(z)-f(y)\,\pi^*_{h,\nu}(\d y, \d z )-\alpha_h(\nu)\bigg)\\
  &\qquad=\sup_{\nu\in \Pc_h^a}\bigg(\int_{\R^d\times \R^d} f(z)-f(y)\,\pi^*_{h,\nu}(\d y, \d z )-\alpha_h(\nu)\bigg),\label{eq.lem.apriori1}
 \end{align}
 where, for $h\in (0,h_0)$ and $a\geq 0$, the set $\Pc_h^a$ consists of all $\nu\in \Pc_h$ with
 \[
 \int_{\R^d\times \R^d}c\bigg(\frac{|y-z|}{h}\bigg)\, \pi^*_{h,\nu}(\d y,\d z)\leq a.
 \]
\end{lemma}

\begin{proof}
Let $L\geq 0$. By Assumption \eqref{cond.P}, there exists some $\gamma\geq 0$ such that
\begin{equation}\label{eq.lem.apriori1.proof}
    c(v)>1+Lv \quad\text{for all }v \in (\gamma,\infty).
\end{equation}
Let $h\in (0,h_0)$ and $f\in \Lip(\R^{d})$ with $\|f\|_{\Lip}\leq L$.\ Then, by Remark \ref{rem.ass.firstorder} c), we may w.l.o.g. assume that $\pi_{h,\mu_h}^*=\mu_h\circ (y,y)^{-1}$, so that 
\begin{equation}\label{eq.reference.zero}
\int_{\R^d\times \R^d} f(z)-f(y)\,\pi_{h,\mu_h}^*(\d y, \d z )-\alpha_h(\mu_h)=\int_{\R^d} f(y)-f(y)\,\mu_h(\d y)=0,
\end{equation}
where, in the first step, we used the fact that $\alpha_h(\mu_h)=0$.\
Since $\mu_h\in \Pc_h$, there exists some $\nu_h\in \Pc_h$ with
\begin{align*}
    0&\leq \sup_{\nu\in \Pc_h}\bigg(\int_{\R^d\times \R^d} f(z)-f(y)\,\pi_{h,\nu}^*(\d y, \d z )-\alpha_h(\nu)\bigg)\\
    &\leq h+\int_{\R^d\times \R^d} f(z)-f(y)\,\pi_{h,\nu_h}^*(\d y, \d z )-\alpha_h(\nu_h)\\
    &\leq h+L\int_{\R^d\times \R^d} |z-y|\,\pi_{h,\nu_h}^*(\d y, \d z )-\alpha_h(\nu_h).
\end{align*}
Using Jensen's inequality and Assumption \eqref{cond.P}, we thus find that
\begin{align*}
    c\bigg(\int_{\R^d\times \R^d} \frac{|z-y|}h\,\pi_{h,\nu_h}^*(\d y, \d z )\bigg)&\leq\int_{\R^d\times \R^d}c\bigg(\frac{|y-z|}{h}\bigg)\, \,\pi^*_{h,\nu_h}(\d y,\d z)\leq \frac{\alpha_h(\nu_h)}{h}\\
    &\leq 1+L\int_{\R^d\times \R^d} \frac{|z-y|}h\,\pi_{h,\nu_h}^*(\d y, \d z ),
\end{align*}
which, by \eqref{eq.lem.apriori1.proof}, implies that $\int_{\R^d\times \R^d} \frac{|z-y|}h\,\pi_{h,\nu_h}^*(\d y, \d z )\leq \gamma$. Hence,
\[
    \int_{\R^d\times \R^d} c\bigg(\frac{|z-y|}{h}\bigg)\, \pi_{h,\nu_h}^*(\d y,\d z)\leq 1+L\int_{\R^d\times \R^d} \frac{|z-y|}h\,\pi_{h,\nu_h}^*(\d y, \d z )\leq 1+L\gamma =:a.
\]
 This shows that every $h$-optimizer of the left-hand side of \eqref{eq.lem.apriori1} belongs to $\Pc_h^a$.\ The proof is complete.
\end{proof}

The previous lemma can be understood as a compactness result in the sense that, for bounded
Lipschitz continuous functions $f\colon \R^d\to \R$, only perturbations with uniformly bounded transport
cost contribute to $I_hf$ for all $h\in (0,h_0)$.\
This allows to extract a limiting convex function $g$, as the following proposition shows.

\begin{proposition}\label{prop.derivative.g_h}
    Assume that Condition \eqref{cond.P} is satisfied.\ Then, for every null sequence in $(0,h_0)$, there exists a subsequence $(h_n)_{n\in \N}$ and a convex function $g\colon \R^d\to \R$ with $g(0)=0\leq g(m)\leq c^*(|m|)$ for all $m\in \R^d$ such that
\begin{equation}\label{eq.existence.g}
        \sup_{|m|\leq L}\bigg|g(m)-\frac{g_{h_n}(m)}{h_n}\bigg|\to 0\quad \text{as }n\to\infty\text{ for all }L\geq0.
    \end{equation}
\end{proposition}

\begin{proof}
By the Cauchy-Schwarz inequality, the definition of $c^*$, and Assumption \eqref{cond.P}, for all $h\in (0,h_0)$ and $m\in \R^d$,
\begin{equation}\label{eq.upperbound.g_h}
g_{h}(m)\leq h c^*(|m|)+\sup_{\nu \in \Pc_h} \Bigg(\int_{\R^{d}\times \R^{d}} hc\bigg(\frac{|z-y|}h\bigg)\,\pi^*_{h,\nu}(\d y,\d z) - \alpha_h(\nu)\Bigg)\leq hc^*(|m|)
\end{equation}
and, by \eqref{eq.reference.zero},
\[
g_{h}(m)\geq \int_{\R^d\times \R^d} \langle m,z-y\rangle\,\pi^*_{h,\mu_h}(\d y,\d z)-\alpha_h(\mu_h)=0.
\]
Moreover, one readily verifies that $g_h\colon \R^d\to \R$ is convex with $g_h(0)=0$ for all $h\in (0,h_0)$.

Now, let $L\geq0$, $\ep>0$, and $a\geq0$ as in Lemma \ref{lem.apriori1}.\ Then, by Assumption \eqref{cond.P}, there exists some $r>0$ such that $c(r)>0$ and $\frac{4Lar}{c(r)}\leq \ep$. Now, let $h\in (0,h_0)$ and $m_1,m_2\in \R^d$ with $|m_i|\leq L$ for $i=1,2$ and $2r|m_1-m_2|\leq \ep$.\ Then, 
\begin{align*}
\bigg|\frac{g_{h}(m_1)}h-\frac{g_{h}(m_2)}h\bigg|&\leq \sup_{\nu \in \Pc_h^a}\int_{\R^{d}\times \R^{d}} |m_1-m_2|\frac{|z-y|}{h}\,\pi_{h,\nu}^*(\d y,\d z)\\
&\leq r|m_1-m_2|+2L \int_{\big\{\frac{|z-y|}{h}>r\big\}} \frac{|z-y|}{h}\,\pi_{h,\nu}^*(\d y,\d z)\\
& \leq \frac{\ep}2+ \frac{2Lr}{c(r)} \int_{\R^{d}\times \R^{d}}c\bigg(\frac{|z-y|}h\bigg)\,\pi_{h,\nu}^*(\d y,\d z)\leq \ep.
\end{align*}
The statement now follows from the Arzel\`a-Ascoli theorem together with a diagonal argument.
\end{proof}

The following theorem is the main result of this section.\ We denote by $ \Cb^{1}(\R^{d}) $ the space of all $ f \in \Cb(\R^{d})$ with bounded and continuous derivative $\nabla f\colon \R^d\to \R^d$.	

\begin{theorem}\label{thm.generator}\
Assume that conditions \eqref{cond.A} and \eqref{cond.P} are satisfied.
\begin{enumerate}
\item[a)]
For all $R\geq 0$ and $f\in \Cb^1(\R^d)$,
\[
    \lim_{h\downarrow 0}\sup_{|x|\leq R}\bigg|\frac{(I_{h}f)(x)-(P_{h}f)(x)-g_{h}\big(\nabla f(x)\big)}{h}\bigg|=0.
\]
 \item[b)] Let $(h_n)_{n\in \N}\subset (0,h_0)$  be a null sequence and $g\colon \R^d\to \R$ such that \eqref{eq.existence.g} holds.\ Then, for all $ f \in \Cb^1(\R^{d})$,
 \begin{equation}\label{eq.thmgenerator1}
	\frac{I_{h_n}f - P_{h_n}f}{h_n} \to g(\nabla f)\quad\text{as }n\to\infty.
 \end{equation}
\end{enumerate}
\end{theorem}
\begin{proof}
Let $f \in \Cb^1(\R^d)$. Then, by \eqref{eq.apriori.lipschitz.thm},
\[
\sup_{h\in (0,h_0)}\frac{\|I_hf-P_hf\|_\infty}{h}\leq c^*\big(\|\nabla f\|_\infty\big),
\]
so that the claim in part b) follows once we have established part a).\ To that end, let $a\geq 0$ as in Lemma \ref{lem.apriori1} with $L:=\|\nabla f\|_\infty$, so that, for all $h\in (0,h_0)$ and $x\in \R^d$,
\begin{align*}
    \bigg|\frac{(I_hf)(x)-(P_hf)(x)}{h}-\frac{g_h\big(\nabla f(x)\big)}{h}\bigg| &\leq  \sup_{\nu\in \Pc_h^a} \frac1h \bigg| \int_{\R^d}f\big(\psi_{h}(x)+z\big)\, \nu(\d z)-(P_hf)(x)\\
    &\qquad\qquad\qquad - \int_{\R^d\times \R^d}\langle\nabla f(x),z-y\rangle \,\pi^*_{h,\nu}(\d y,\d z)\bigg|.
\end{align*}
Using the fundamental theorem of calculus, for all $\nu\in \Pc_h^a$,
\begin{align*}
    \frac1h \bigg|\int_{\R^d}f\big(\psi_{h}(x)&+z\big)\, \nu(\d z)-(P_hf)(x)
    - \int_{\R^d\times \R^d}\langle\nabla f(x),z-y\rangle \,\pi^*_{h,\nu}(\d y,\d z)\bigg|\\
    &\leq \int_{\R^d\times \R^d}\int_0^1\big|\nabla f\big(\psi_h(x)+y+s(z-y)\big)-\nabla f(x)\big|\frac{|z-y|}{h} \,\d s\, \pi_{h,\nu}^*(\d y,\d z).
\end{align*}	
We split the last integral in three parts and estimate each one of them separately.\ Let $R\geq 0$ and $\ep>0$. Then, there exists $r>0$ such that $\frac{6Lar}{c(r)}\leq \ep$.\ Moreover, by \eqref{eq.ass.psi_h}, there exist $\overline h\in (0,h_0]$ and $\de>0$ such that
\begin{equation}\label{eq.unifcont.gradient}
    \sup_{|x|\leq R}\big|\nabla f\big(\psi_h(x)+u\big)-\nabla f(x)\big|\leq \frac{\epsilon}{3r}\quad\text{for all }h\in \big(0,\overline h\big)\text{ and }u\in \R^d\text{ with }|u|\leq 2\delta.
\end{equation} 
Last but not least, since $\mu_h\to \de_0$ as $h\downarrow0$ by Assumption \eqref{cond.A}, after a potential modification of $\overline h\in (0,h_0]$, we may w.l.o.g.\ assume that $ r\overline h\leq \de$ and
\begin{equation}\label{eq.cond.M.proof.main1}
6Lr\mu_h\big(\big\{y\in \R^d\,\big|\, |y|>\de\big\}\big)\leq\ep\quad \text{for all }h\in \big(0,\overline h\big).
\end{equation}
Then, for all $h\in \big(0,\overline h\big)$, $\nu\in \Pc_h^a$, and $x\in \R^d$ with $|x|\leq R$, by \eqref{eq.unifcont.gradient},
\begin{align*}
    \int_{\{|y|\leq \delta\}\cap \big\{\frac{|z-y|}h\leq r\big\}}\int_0^1 \Big|\nabla f\big(\psi_h(x)+y+s(z-y)\big)-\nabla f(x)\Big|&\frac{|z-y|}{h}\,\d s\,\pi_{h,\nu}^*(\d y,\d z)\leq \frac{\ep}{3},
\end{align*}
by \eqref{eq.cond.M.proof.main1},
\begin{align*}
    \int_{\{|y|> \delta\}\cap \big\{\frac{|z-y|}h\leq r\big\}}&\int_0^1 \Big|\nabla f\big(\psi_h(x)+y+s(z-y)\big)-\nabla f(x)\Big|\frac{|z-y|}{h}\,\d s\,\pi_{h,\nu}^*(\d y,\d z)\\
    &\leq 2Lr\mu_h\big(\big\{y\in \R^d\,\big|\, |y|>\de\big\}\big)\leq\frac{\ep}3,
\end{align*}
and
\begin{align*}
    \int_{\big\{\frac{|z-y|}h> r\big\}}&\int_0^1 \Big|\nabla f\big(\psi_h(x)+y+s(z-y)\big)-\nabla f(x)\Big|\frac{|z-y|}{h}\,\d s\,\pi_{h,\nu}^*(\d y,\d z)\\
    &\leq 2L \int_{\big\{\frac{|z-y|}h> r\big\}} \frac{|z-y|}{h}\,\pi_{h,\nu}^*(\d y,\d z)\leq \frac{2Lr}{c(r)}\int_{\R^d\times \R^d} c\bigg(\frac{|z-y|}{h}\bigg)\,\pi_{h,\nu}^*(\d y,\d z)\\
    & \leq \frac{2Lar}{c(r)}\leq\frac{\ep}3.
\end{align*}
The proof is complete.
\end{proof}

We conclude this section with a series of examples for penalizations that satisfy Assumption \eqref{cond.P} and illustrate how the quantity $g_h$ for $h\in (0,h_0)$ captures transport-induced displacements under the cost constraint $\alpha_h$ in concrete settings. In all cases, the function $g$ from the previous theorem can be explicitly computed and $\frac{g_{h}(m)}{h}\to g(m)$ as $h\downarrow 0$ for all $m\in \R^d$ (even $\frac{g_{t}(m)}{t}= g(m)$ for all $t>0$ and $m\in \R^d$).

\begin{example}[Optimal transport penalization]\label{ex.fop}\
Let $\phi\colon \R^d\to [0,\infty)$ be measurable with $\phi(0)=0$ and $\lim_{|u|\to \infty}\frac{\phi(u)}{|u|}=\infty$. For all $t>0$, let 
    \[
        \alpha_t(\nu):=\inf_{\pi\in \Cpl(\mu_t,\nu)} \int_{\R^d\times \R^d} t\phi\bigg(\frac{z-y}{t}\bigg)\, \pi(\d y,\d z)\in [0,\infty].
    \]
    In order to show that Assumption \eqref{cond.P} is satisfied, let
    \[
        c(v):=\inf_{|u|\geq v} \phi(u)\quad\text{for all }v\geq 0.
    \]
    Then, by assumption, $c\colon [0,\infty)\to [0,\infty)$ is nondecreasing with
    \[
        \frac{c(v)}{v}\geq \inf_{|u|\geq v}\frac{\phi(u)}{|u|}\to \infty\quad \text{as }v\to \infty.
    \]
    Moreover, $\phi(u)\geq c(|u|)$ for all $u\in \R^d$, so that
    \begin{align*}
        \alpha_t(\nu)&= \inf_{\pi\in \Cpl(\mu_t,\nu)} \int_{\R^d\times \R^d} t\phi\bigg(\frac{z-y}{t}\bigg)\, \pi(\d y,\d z)\\
        &\geq \inf_{\pi\in \Cpl(\mu_t,\nu)} \int_{\R^d\times \R^d} tc\bigg(\frac{|z-y|}{t}\bigg)\, \pi(\d y,\d z)
    \end{align*}
    for all $t>0$ and $\nu\in \Pc(\R^d)$.\ For $t>0$, $b\in \R^d$, and $\nu= \mu_t\circ (\,\cdot+ b)^{-1}$, we choose the coupling $\pi_{t,\nu}^* := \mu_t \circ (\,\cdot\,, \,\cdot+ b)^{-1}$, and show that $$\frac{g_t(m)}{t}= \phi^*(m):=\sup_{u\in \R^d}\big(\langle m,u\rangle-\ph(u)\big)\in [0,\infty)$$ for all $t>0$ and $m\in \R^d$.\ By definition of the penalty function $\al$, similar as in  \eqref{eq.upperbound.g_h}, it follows that
    \[
    \frac{g_t(m)}{t}\leq \phi^*(m)\quad \text{for all }t>0\text{ and }m\in \R^d.
    \] 
    On the other hand, choosing $\nu = \mu_t \circ (\,\cdot + b)^{-1}$ for arbitrary $b\in \R^d$,
    \[
        \frac{g_t(m)}{t}\geq \sup_{b\in \R^d} \left(\frac{1}{t}\langle m,b\rangle -\phi\Big(\frac{b}{t}\Big)\right)=\phi^*(m),
    \]
which shows that $\frac{g_t(m)}{t}=\phi^*(m)$ for all $t>0$ and $m\in \R^d$.
\end{example}

\begin{example}[Wasserstein penalization]\label{ex.wassersteinfop}
Let $\phi\colon [0,\infty)\to [0,\infty]$ be convex and lower semicontinuous with $\phi(0)=0 $ and $\phi(v)\neq 0$ for some $v\in (0,\infty)$.\ Then, $\phi$ is nondecreasing and continuous on $ \operatorname{dom}(\phi):= \{v \in [0,\infty)\,|\, \phi(v) < \infty\}$.\ Further, let $ p \in (1,\infty)$, and assume that the map $v \mapsto \phi(v^{1/p})$ is convex. Since $\phi \not \equiv 0$, this implies that $\liminf_{v \to \infty} \frac{{\phi}(v)}{v^p}>0$. Let
\[  
    \alpha_t(\nu):=t\phi\bigg(\frac{\Wc_p(\mu_t,\nu)}t\bigg) \quad\text{for all }\nu\in \Pc(\R^d),
\]
where $\Wc_p$ denotes the Wasserstein $p$-distance between $\mu_t$ and $\nu\in \Pc(\R^d)$, i.e.,
\begin{equation}\label{eq.def.wass.p}
 \Wc_p(\mu_t,\nu):=\bigg(\inf_{\pi \in \Cpl(\mu_t,\nu)} \int_{\R^{d} \times \R^{d}}|z-y|^p\,\pi(\d y,\d z)\bigg)^{1/p}\in [0,\infty],
\end{equation}
 and we set $\phi(\infty):=\infty$.\ Since $ \phi(0)= 0$, it follows that $\alpha_t(\mu_t)=0$ for all $t>0$. 
Since $\liminf_{v \to \infty} \frac{\phi(v)}{v^p}>0$ and $\phi$ is nondecreasing, for all $ \epsilon >0 $, there exists some $a \geq 0$ such that
\[
    \phi(v) \geq \epsilon v^p, \quad \text{for all }v \geq a.
\]
This yields that for all $ v \geq 0 $ we have 
\[
    \phi(v) = \phi(v)- \epsilon v^p + \epsilon v^p \geq \left(\min_{u \in [0,a]}\phi(u) - \epsilon u^p\right) + \epsilon v^p = \epsilon v^p - c, 
\]
where $c:= - \min_{u \in [0,a]}(\phi(u) - \epsilon u^p) \geq 0$. Let $ t > 0 $ and $ \nu \in \Pc_t$.\ Then, 
\begin{align*}
	\frac{\alpha_t(\nu)}{t} &= \phi\left(\frac{\Wc_p(\mu_t,\nu)}{t}\right)
	\geq \epsilon \bigg(\frac{\Wc_p(\mu_t,\nu)}{t}\bigg)^p-c\\
	&= \inf_{\pi \in \Cpl(\mu_t,\nu)} \Bigg(\int_{\R^{d} \times \R^{d}} \epsilon\bigg(\frac{|y-z|}{t}\bigg)^p-c \,\pi(\d y,\d z)\Bigg).
\end{align*}
Defining $ c(v):= \epsilon v^p-c$ for all $v\geq 0$, shows that Assumption \eqref{cond.P} is satisfied. For all $t>0$ and $\nu\in \Pc_t$, let $\pi_{t,\nu}^*$ be an optimal coupling in the definition of the Wasserstein $p$-distance \eqref{eq.def.wass.p}. Let $t >0$ and $m \in \R^{d}$.\ Then,
\begin{align*}
	\frac{g_t(m)}{t}&=\sup_{\nu \in \Pc_t}\Bigg(\int_{\R^d\times \R^d}\bigg\langle m,\frac{z-y}{t}\bigg\rangle\,\pi^*_{t,\nu}(\d y, \d z) - \phi\bigg(\frac{\Wc_p(\mu_t,\nu)}{t}\bigg)\Bigg)\\
    &
	\leq \sup_{\nu \in \Pc_t}\Bigg(|m| \frac{\Wc_p(\mu_t,\nu)}{t} - \phi\bigg(\frac{\Wc_p(\mu_t,\nu)}{t}\bigg)\Bigg)\leq  \phi^*(|m|).
\end{align*}
Since, by assumption, the map $[0,\infty)\to [0,\infty],\; v \mapsto \phi(v^{1/p})$ is convex, Jensen's inequality implies that	
\begin{equation}\label{ineq.wasserstein} 
	\phi\bigg(\frac{\Wc_p(\mu_t,\nu)}{t}\bigg)\leq \int_{\R^{d} \times \R^{d}} \phi\bigg(\frac{|y-z|}{t}\bigg)\, \pi^*_{t,\nu}(\d y,\d z)\quad \text{for all }\nu\in \Pc_t.
\end{equation}
Since the map $x\mapsto |x|^p$ is strictly convex, by Remark \ref{rem.ass.firstorder} c) and Jensen's inequality, for $t>0$ and $\nu=\mu_t\circ (\,\cdot+b)^{-1}$ with $b\in \R^d$, the unique optimal coupling for the Wasserstein $p$-distance is $\pi_{t,\nu}^*=\mu_t\circ (\,\cdot\,,\,\cdot+b)^{-1}$, so that that
\begin{align*}
	\frac{g_t(m)}{t}&\geq \sup_{b\in \R^d} \Bigg(\bigg\langle m,\frac{b}{t}\bigg\rangle- \phi\bigg(\frac{|b|}{t}\bigg)\Bigg)= \phi^*(|m|).
\end{align*}
This shows that $\frac{g_t(m)}t=\phi^*(|m|)$ for all $t>0$ and $m\in \R^d$.
\end{example}

\begin{example}[Drift control penalization]\label{ex.optimalcontrol.fop}
    Let $\phi\colon \R^d\to [0,\infty)$ be convex and lower semicontinuous with $\phi(0)=0$ and $\lim_{|u|\to \infty}\frac{\phi(u)}{|u|}=\infty$, $(\Omega,\mathcal F,\P)$ be a probability space, and $(Y_t)_{t\geq 0}$ be a family of random variables on $(\Omega,\mathcal F,\P)$ with $Y_t\sim \mu_t$ for all $t>0$.\ Let $\mathbb F=(\mathcal F_t)_{t\geq 0}$ be a filtration on $(\Omega,\mathcal F)$. Denote by $\mathcal A$ the set of all $\mathbb F$-progressively measurable $\R^d$-valued stochastic processes $\beta=(\beta_t)_{t\geq 0}$ with $\E\big(\int_0^t \phi(\beta_s)\, \d s\big)<\infty$ for all $t\geq 0$.\ For $t>0$ and $\beta\in \mathcal A$, let $\mu_t^\beta\in \Pc(\R^d)$ denote the law of $Y_t+\int_0^t\beta_s\, \d s$. For $t>0$ and $\nu\in \Pc(\R^d)$, let
    \[
        \alpha_t(\nu):=\inf \bigg\{ \E\bigg[\int_0^t \phi(\beta_s)\, \d s\bigg]\bigg|\, \beta\in \mathcal A,\, \mu_t^\beta=\nu\bigg\},
    \]
    where we use the convention $\inf \emptyset=\infty$.\ Then, $\Pc_t=\big\{\mu_t^\beta\,\big|\, \beta\in \mathcal A\big\}$ for all $t>0$
    and
    \begin{align*}
        \inf_{\pi\in \Cpl(\mu_t,\nu)}\int_{\R^d\times \R^d}t\phi\bigg(\frac{z-y}{t}\bigg)\, \pi(\d y,\d z)&\leq \inf\bigg\{ \E\bigg[t\phi\bigg(\frac1t \int_0^t \beta_s\, \d s\bigg)\bigg]\bigg|\, \beta\in \mathcal A,\, \mu_t^\beta=\nu\bigg\}\\
        & \leq \inf\bigg\{ \E\bigg[ \int_0^t \phi(\beta_s)\, \d s\bigg]\bigg|\, \beta\in \mathcal A,\, \mu_t^\beta=\nu\bigg\}\\
        &=\alpha_t(\nu)\quad\text{for all }t>0\text{ and }\nu\in \Pc_t,
    \end{align*}
    where the first inequality follows by choosing the coupling $\pi = \mathbb{P} \circ (Y_t, Y_t + \int_0^t \beta_s\,\mathrm{d}s)^{-1}$ for  $\beta\in \mathcal A$ with $\mu_t^\beta=\nu$ and the second inequality follows from Jensen's inequality.
    
    Choosing $c$ similarly as in Example \ref{ex.fop} in such a way that $c(v)<\infty$ for all $v\geq0$,\footnote{This can, for example, be achieved by replacing $\ph$ with $\tilde\ph(u):=\min\{\ph(u),|u|^2\}$ for $u\in \R^d$ in the definition of the function $c$ in Example \ref{ex.fop}.}  
    we have therefore shown that Assumption~\eqref{cond.P} is satisfied.\ Choosing the coupling $\pi^*_{t,\nu} := \mathbb{P} \circ (Y_t, Y_t +  b)^{-1}$ for $t>0$ and $\nu=\mathbb{P} \circ (Y_t +  b)^{-1}$ with $b\in \R^d$, it follows that
    \[
        \frac{g_t(m)}t\geq \sup_{b\in \R^d} \left(\frac{1}{t}\langle m,b\rangle -\phi\Big(\frac{b}{t}\Big)\right)=\phi^*(m).
    \]
Hence, using a similar estimate as in \eqref{eq.upperbound.g_h}, cf.\ Example \ref{ex.fop}, $\frac{g_t(m)}{t}=\phi^*(m)$ for all $t>0$ and $m\in \R^d$.
\end{example}

\section{Second-order perturbation}\label{sec:second order}

In this section, we consider the case of second-order penalizations.\ Again, we consider the family $(I_t)_{t>0}$, given by \eqref{eq.operatorI}.\ However, this time with a family $(\alpha_t)_{t > 0}$ of penalization functions $\Pc(\R^d)\to [0,\infty]$ satisfying the following new assumption.
\begin{enumerate}
    \item[(P')]\makeatletter\def\@currentlabel{P'}\makeatother\label{cond.Pprime} There exist $h_0\in(0,\infty)$ and a nondecreasing function $c\colon [0,\infty)\to \R$ with $$\lim_{v\to \infty}\frac{c(v)}{v}=\infty$$ such that
    \[
        \alpha_h(\nu)\geq \inf_{\pi\in \mart(\mu_h,\nu)}\int_{\R^d\times \R^d}hc\bigg(\frac{|y-z|^2}{2h}\bigg)\,\pi(\d y,\d z)\quad\text{for all }h\in(0,h_0)\text{ and }\nu\in \Pc(\R^d),
    \]
    where we use again the convention $\inf \emptyset=\infty$.\ Moreover, $\alpha_h(\mu_h)=0$ for all $h\in (0,h_0)$.
\end{enumerate}	

Replacing again $c$ by $c^{**}\leq c$, see Lemma \ref{lem.conjugate} b), we may w.l.o.g.\ assume that the function $c$ is convex, so that $c$ is continuous and $(0,\infty)\to \R,\; v\mapsto \frac{c(v)}{v}$ is nondecreasing since, by Assumption \eqref{cond.Pprime}, $c(0)\leq 0$.\footnote{Since $c$ is nondecreasing,
     $c(0)\leq \inf_{\pi\in \mart(\mu_h,\mu_h)}\int_{\R^d\times \R^d} h c\big(\frac{|z-y|^2}{2h}\big)\,\pi(\d y,\d z)\leq \alpha_h(\mu_h)=0$ for $h\in (0,h_0)$.}\ 
     We define the convex conjugate of $c$ by
\[
c^\ast(w):=\sup_{v\ge 0}\big(vw-c(v)\big) \quad\text{for all } w\ge 0.
\]
Then, by Lemma \ref{lem.conjugate} a), the supremum is again finite for every $w\ge 0$, so that $c^\ast(w)\in [0,\infty)$ for all $w\geq 0$.

    Again, probability measures $\nu\in \Pc(\R^d)$ with $\alpha_t(\nu)=\infty$ do not affect the value
of the supremum in \eqref{eq.operatorI}, so that the supremum may equivalently
be taken over the effective domain
\[
\Pc_t := \big\{\nu\in\Pc(\R^d)\,\big|\,\alpha_t(\nu)<\infty\big\}\quad \text{for all }t>0.
\]
This convention will be used throughout this section.\ In particular, $\mart(\mu_h,\nu)\neq \emptyset$ for all $h\in (0,h_0)$ and $\nu\in \Pc_h$.\ We start with a series of observations.

\begin{remark}\label{rem.ass.secondorder}\
Assume that Condition \eqref{cond.Pprime} is satisfied.
\begin{enumerate}  
    \item[a)] Let $f\colon \R^d\to \R$ be differentiable with Lipschitz continuous gradient $\nabla f\colon \R^d\to \R^d$ and $h\in (0,h_0)$.\ Then, the assumptions on the family $(\mu_t)_{t>0}$ do, in general, neither ensure that $f$ is $\mu_t$-integrable for $t>0$ nor that
    \[
        \int_{\R^d\times \R^d}\big|\langle \nabla f(y),z-y\rangle \big|\,\pi(\d y,\d z)<\infty
    \]
    for $\pi\in \mart(\mu_t,\nu)$ and $t>0$.\ However, by the fundamental theorem of calculus,
    \begin{align}
        \notag \big|f(z)-f(y)-\langle \nabla f(y),z-y\rangle\big|&=\bigg|\int_0^1 \big\langle\nabla f\big(y+s(z-y)\big)-\nabla f(y),z-y\big\rangle\, \d s\bigg|\\
        &\leq \|f\|_{\Lip,1}\frac{|z-y|^2}{2}\quad \text{for all }y,z\in \R^d,\label{eq.bound.fund.calc}
    \end{align}
    where
    \begin{equation}
        \label{eq.norm.lip.1}
        \|f\|_{\Lip,1}:=\sup_{\substack{y,z\in \R^d\\ y\neq z}}\sup_{s\in (0,1]}\frac{\big|\langle \nabla f\big(y+s(z-y)\big)-\nabla f(y),z-y\rangle\big| }{s|z-y|^2}\leq \|\nabla f\|_{\Lip}
    \end{equation}
    Hence, for all $\nu \in \mathcal{P}_h$ and any martingale coupling $\pi \in \mart(\mu_h,\nu)$ with
       \begin{equation}\label{eq.finite.cost2}
         \int_{\R^d\times \R^d} hc\bigg(\frac{|z-y|^2}{2h}\bigg)\,\pi(\d y,\d z)<\infty,
        \end{equation}
        we find that
        \begin{align}
          \notag &\int_{\R^d\times \R^d} \big|f(z)-f(y)-\langle \nabla f(y),z-y\rangle\big|\,\pi(\d y,\d z)-\int_{\R^d\times \R^d} hc\bigg(\frac{|z-y|^2}{2h}\bigg)\,\pi(\d y,\d z)\\
           &\qquad\quad\leq h\int_{\R^d\times \R^d} \|f\|_{\Lip,1}\frac{|z-y|^2}{2h}-c\bigg(\frac{|z-y|^2}{2h}\bigg)\,\pi(\d y,\d z)\leq h c^*(\|f\|_{\Lip,1})<\infty, \label{eq.mart.wass.finite}
        \end{align}
        so that the mapping $$\R^d\times \R^d \to \R,\quad (y,z)\mapsto f(z)-f(y)-\langle \nabla f(y),z-y\rangle$$ is $\pi$-integrable.\ Choosing $f(x)=\frac12|x|^2$ for $x\in \R^d$, it follows that
        \begin{equation}\label{eq.apriori.mart.coupling}
         \int_{\R^d\times \R^d} \frac{|z-y|^2}{2}\, \pi(\d y,\d z)\leq hc^*(1)+\int_{\R^d\times \R^d}hc\bigg(\frac{|z-y|^2}{2h}\bigg)\,\pi(\d y,\d z)<\infty
        \end{equation}
        for all $\nu\in \Pc_h$ and any martingale coupling $\pi\in  \mart(\mu_h,\nu)$ with \eqref{eq.finite.cost2}.\
        Moreover, by \eqref{eq.mart.wass.finite}, for all $\nu\in \Pc_h$ and any martingale coupling $\pi\in  \mart(\mu_h,\nu)$ with 
        \begin{equation}\label{eq.almost.mart.opt.coupling}
   \int_{\R^d\times \R^d} hc\bigg(\frac{|z-y|^2}{2h}\bigg)\, \pi(\d y,\d z)\leq \al_h(\nu),
\end{equation}
       it holds
\begin{equation}\label{eq.apriori.lipschitz.mart}
         \int_{\R^d\times \R^d} \big|f(z)-f(y)-\langle \nabla f(y),z-y\rangle\big|\,\pi(\d y, \d z )-\alpha_h(\nu)\leq hc^*\big(\|f\|_{\Lip,1}\big).
          \end{equation}
\item[b)] Since $c$ is continuous and nondecreasing, by \eqref{eq.apriori.mart.coupling}, the de la Vall\'ee Poussin Lemma \cite[Theorem 4.5.9]{Boga2007}, the Lebesgue-Vitali Theorem \cite[Theorem 4.5.4]{Boga2007}, and \cite[Lemma 4.4 and Definition 6.8]{villani2008optimal}, for all $h\in (0,h_0)$ and $\nu\in \Pc_h$, there exists an optimal martingale coupling $\pi^*_{h,\nu}\in \mart(\mu_h,\nu)$ with 
    \begin{align}
        \notag \int_{\R^d\times \R^d} hc\bigg(\frac{|z-y|^2}{2h}\bigg)\,\pi^*_{h,\nu}(\d y,\d z)&=\inf_{\pi\in \mart(\mu_h,\nu)}\int_{\R^d\times \R^d} hc\bigg(\frac{|z-y|^2}{2h}\bigg)\,\pi(\d y,\d z)\\
        &\leq \alpha_h(\nu)<\infty. \label{eq.mart.opt.coupling}
        \end{align}  
        Hence, by \eqref{eq.def.mart.coupling} and \eqref{eq.apriori.lipschitz.mart}, for any bounded differentiable function $f\colon \R^d\to \R$ with bounded Lipschitz continuous gradient $\nabla f\colon \R^d\to \R^d$ and $h\in (0,h_0)$,
          \begin{align}
  \notag        0&\leq (I_h f)(x)-(P_h f)(x)\\
\notag &\qquad \quad =
\sup_{\nu\in\Pc_h}
\bigg(
\inf_{\pi\in\mart(\mu_h,\nu)}
\int_{\R^d\times \R^d} f\big(\psi_h(x)+z\big)-f\big(\psi_h(x)+y\big)\,\pi(\d y,\d z)
-\alpha_h(\nu)
\bigg) \label{eq.apriori.mart.lipschitz.thm}\\
&\qquad\quad \le hc^*\big(\| f\|_{\Lip,1}\big) \quad \text{for all }x\in \R^d.
\end{align}
\item[c)] Let $f\colon \R^d\to \R$ be differentiable with Lipschitz continuous gradient $\nabla f\colon \R^d\to \R^d$, $h\in (0,h_0)$, and  $\nu\in \Pc_h$. Then, for all $\pi\in\mart(\mu_h,\nu)$ with
\begin{equation}\label{eq.finite.second.moment}
\int_{\R^d\times\R^d} |z-y|^2\,\pi(\d y,\d z)<\infty,
\end{equation}
the integral
\begin{equation}\label{eq.int.secondorder}
\int_{\R^d\times\R^d} f(z)-f(y)-\langle \nabla f(y), z-y\rangle \,\pi(\d y,\d z)
\end{equation}
is well-defined and independent of the choice of $\pi\in\mart(\mu_h,\nu)$ among all martingale couplings satisfying the integrability condition \eqref{eq.finite.second.moment}.\ Indeed, by \eqref{eq.bound.fund.calc}, the integral in \eqref{eq.int.secondorder} is well-defined for all $\pi\in\mart(\mu_h,\nu)$ with \eqref{eq.finite.second.moment}. Now, let $\pi_1,\pi_2\in\mart(\mu_h,\nu)$ with \eqref{eq.finite.second.moment} and $\ph\in \Ccinf(\R^d)$ with $\ph\geq 0$ and $\int_{\R^d} \ph(u)\,\d u=1$. Then, for all $u\in \R^d$,
\begin{align*}
 \int_{\R^d\times \R^d}& \ph(z+u)-\ph(y+u)-\langle \nabla \ph(y+u),z-y\rangle \,\pi_1(\d y,\d z)\\
 &=\int_{\R^d\times \R^d} \ph(z+u)-\ph(y+u)\,\pi_1(\d y,\d z)\\
&=\int_{\R^d} \ph(z+u)\,\nu(\d z)-\int_{\R^d} \ph(y+u)\,\mu_h(\d y)\\
&=\int_{\R^d\times \R^d} \ph(z+u)-\ph(y+u)\,\pi_2(\d y,\d z)\\
&=\int_{\R^d\times \R^d} \ph(z+u)-\ph(y+u)-\langle \nabla \ph(y+u),z-y\rangle \,\pi_2(\d y,\d z).
\end{align*}
Since 
\begin{align*}
 \int_{\R^d\times \R^d}\int_{\R^d} \ph(u)&\big|f(z+u)-f(y+u)-\langle \nabla f(y+u),z-y\rangle\big| \,\d u\,\pi_i(\d y,\d z)\\
 &\leq \|f\|_{\Lip,1} \int_{\R^d\times\R^d} \frac{|z-y|^2}2 \,\pi_i(\d y,\d z)<\infty\quad \text{for }i=1,2,
\end{align*}
by the Fubini-Tonelli theorem,
\begin{align*}
 \int_{\R^d\times \R^d}& (f\ast\ph)(z)-(f\ast\ph)(y)-\big\langle \nabla (f\ast \ph)(y),z-y\big\rangle \,\pi_1(\d y,\d z)\\
 &=\int_{\R^d} f(u)\int_{\R^d\times \R^d} \ph(z+u)-\ph(y+u)-\langle \nabla \ph(y+u),z-y\rangle \,\pi_1(\d y,\d z)\,\d u\\
 &=\int_{\R^d} f(u)\int_{\R^d\times \R^d} \ph(z+u)-\ph(y+u)-\langle \nabla \ph(y+u),z-y\rangle \,\pi_2(\d y,\d z)\,\d u\\
 &= \int_{\R^d\times \R^d} (f\ast\ph)(z)-(f\ast\ph)(y)-\big\langle \nabla (f\ast \ph)(y),z-y\big\rangle \,\pi_2(\d y,\d z).
\end{align*}
Since $\|f\ast \ph\|_{\Lip,1}\leq \| f\|_{\Lip,1}$, using the dominated convergence theorem and Friedrichs mollifier,
the claim follows.
\end{enumerate}
\end{remark}

Again, we aim to identify the second-order limit of the rescaled nonlinear perturbations.
As in the previous section, for each $h\in (0,h_0)$, we introduce an auxiliary function $G_h$ which captures the
maximal local covariance displacement induced by perturbations of the
reference measure $\mu_h$, penalized by the transport cost $\alpha_h$.

Throughout, let $\sym$ denote the set of all symmetric matrices $A\in\R^{d\times d}$ endowed with the norm
    \begin{equation}\label{eq.matrixnorm}
 	    |A|:= \sup_{|x| =1} \big|\langle Ax, x \rangle\big|\quad\text{for }A\in \sym.
    \end{equation}
    Throughout the remainder of this section, for all $h\in (0,h_0)$ and all $\nu\in \Pc_h$, we fix an arbitrary martingale coupling $\pi_{h,\nu}^*\in \mart(\mu_h,\nu)$ with \eqref{eq.finite.second.moment} and consider the map $G_h\colon \sym\to \R$, given by 
    \begin{equation}\label{eq.def.G_h}
        G_h(A):=\sup_{\nu\in \Pc_h}\bigg(\int_{\R^d\times \R^d} \frac12\big\langle A(z-y),z-y\big\rangle\, \pi_{h,\nu}^*(\d y,\d z)-\alpha_h(\nu)\bigg)\quad \text{for all }A\in \sym.
    \end{equation}
Choosing $f(x)=\frac12\langle Ax,x\rangle$ for $x\in \R^d$ with $A\in \sym$, by Remark \ref{rem.ass.secondorder} c), the function $G_h$ is well-defined and the definition of $G_h$ is independent of the choice of the martingale coupling $\pi_{h,\nu}^*\in \mart(\mu_h,\nu)$ among all couplings satisfying the integrability condition \eqref{eq.finite.second.moment} for all $h\in (0,h_0)$ and $\nu\in \Pc_h$.\ In particular, we may assume w.l.o.g.\ that 
\[
   \int_{\R^d\times \R^d} hc\bigg(\frac{|z-y|^2}{2h}\bigg)\, \pi_{h,\nu}^*(\d y,\d z)\leq \al_h(\nu),
\]
see Remark \ref{rem.ass.secondorder} a), i.e., $\pi_{h,\nu}^*$ satisfies \eqref{eq.almost.mart.opt.coupling} in Remark \ref{rem.ass.secondorder} b) for all $h\in (0,h_0)$ and $\nu\in \Pc_h$.

The following auxiliary result will again play a crucial role in the proof of the main theorem of this section.

\begin{lemma}\label{lem.apriori2} Assume that Condition \eqref{cond.Pprime} is satisfied.\
 Then, for all $L\geq 0$, there exists some $a\geq 0$ such that, for all $h\in (0,h_0)$ and all differentiable functions $f\colon \R^d\to \R$ with $\|\nabla f\|_{\Lip}\leq L$,
 \begin{align}
  \notag\sup_{\nu\in \Pc_h}\bigg(&\int_{\R^d\times \R^d} f(z)-f(y)-\big\langle \nabla f(y),z-y\big\rangle\,\pi_{h,\nu}^*(\d y, \d z )-\alpha_h(\nu)\bigg)\\
  &=\sup_{\nu\in \Pc_h^a}\bigg(\int_{\R^d\times \R^d} f(z)-f(y)-\big\langle \nabla f(y),z-y\big\rangle\,\pi_{h,\nu}^*(\d y, \d z )-\alpha_h(\nu)\bigg),\label{eq.lem.apriori2}
 \end{align}
 where, for $h\in (0,h_0)$ and $a\geq 0$, the set $\Pc_h^a$ consists of all $\nu\in \Pc_h$ with
 \[
 \int_{\R^d\times \R^d}c\bigg(\frac{|y-z|^2}{2h}\bigg)\, \pi_{h,\nu}^*(\d y,\d z)\leq a.
 \]
\end{lemma}

\begin{proof}
Let $L\geq 0$.\ Then, by assumption on $c$, there exists some $\gamma\geq 0$ such that
\begin{equation}\label{eq.lem.apriori2.proof}
    c(v)>1+Lv \quad\text{for all }v \in (\gamma,\infty).
\end{equation}
Now, let $h\in (0,h_0)$ and $f\colon\R^d\to \R$ differentiable with $\|\nabla f\|_{\Lip}\leq L$.\ Since $\mu_h\in \Pc_h$ with $\alpha_h(\mu_h)=0$, there exists some $\nu\in \Pc_h$ with
\begin{align*}
    0&\leq \sup_{\nu\in \Pc_h}\bigg(\int_{\R^d\times \R^d} f(z)-f(y)-\big\langle \nabla f(y),z-y\big\rangle\,\pi_{h,\nu}^*(\d y, \d z )-\alpha_h(\nu)\bigg)\\
    &\leq h+\int_{\R^d\times \R^d} f(z)-f(y)-\big\langle \nabla f(y),z-y\big\rangle\,\pi_{h,\nu}^*(\d y, \d z )-\alpha_h(\nu)\\
    &\leq h+L\int_{\R^d\times \R^d}\frac{|z-y|^2}2\, \pi_{h,\nu}^*(\d y,\d z)-\alpha_h(\nu).
\end{align*}
Using Jensen's inequality and Assumption \eqref{cond.Pprime}, we thus find that
\begin{align*}
    c\bigg(\int_{\R^d\times \R^d}\frac{|z-y|^2}{2h}\, \pi_{h,\nu}^*(\d y,\d z)\bigg)&\leq \int_{\R^d\times \R^d}c\bigg(\frac{|y-z|^2}{2h}\bigg)\, \,\pi_{h,\nu}^*(\d y,\d z)\\
    &\leq \frac{\alpha_h(\nu)}{h}\leq 1+L\int_{\R^d\times \R^d}\frac{|z-y|^2}{2h}\, \pi_{h,\nu}^*(\d y,\d z),
\end{align*}
which, by \eqref{eq.lem.apriori2.proof}, implies that $\int_{\R^d\times \R^d}\frac{|z-y|^2}{2h}\, \pi_{h,\nu}^*(\d y,\d z)\leq \gamma$.\ Hence,
\[
\int_{\R^d\times \R^d} c\bigg(\frac{|z-y|^2}{2h}\bigg)\, \pi_{h,\nu}^*(\d y,\d z)\leq 1+L\int_{\R^d\times \R^d}\frac{|z-y|^2}{2h}\, \pi_{h,\nu}^*(\d y,\d z)\leq 1+L\gamma =:a.
\]
The proof is complete.
\end{proof}

Again, the previous lemma can be understood as a compactness result, which allows to extract a limiting convex function $G$, which is nondecreasing, i.e.,
\[
 G(A_1)\leq G(A_2)\quad \text{for all }A_1,A_2\in \sym\text{ s.t.\ }A_2-A_1\text{ is positive semidefinite,} 
\]
as the following proposition shows.

\begin{proposition}\label{prop.derivative.G_h}
     Assume that Condition \eqref{cond.Pprime} is satisfied.\ Then, for every null sequence in $(0,h_0)$, there exists a subsequence $(h_n)_{n\in \N}$ and a convex nondecreasing function $G\colon \sym\to \R$ with $G(0)=0\leq G(A)\leq c^*(|A|)$ for all $A\in \sym$ such that
	\begin{equation}\label{eq.existence.G}
        \sup_{|A|\leq L}\bigg|G(A)-\frac{G_{h_n}(A)}{h_n}\bigg|\to 0\quad\text{as }n\to\infty\text { for all } L\geq 0,
    \end{equation}
\end{proposition}

\begin{proof}
By the definition of the norm on $\sym$, the definition of $c^*$, and Assumption \eqref{cond.Pprime}, for all $h\in (0,h_0)$ and $A\in \sym$,
\begin{equation}\label{eq.upperbound.G_h}
G_{h}(A)\leq h c^*(|A|)+\sup_{\nu \in \Pc_h} \Bigg(\int_{\R^{d}\times \R^{d}} hc\bigg(\frac{|z-y|^2}{2h}\bigg)\,\pi^*_{h,\nu}(\d y,\d z) - \alpha_h(\nu)\Bigg)\leq hc^*(|A|)
\end{equation}
and, choosing $\pi^*_{h,\mu_h}=\mu_h\circ (y,y)^{-1}$ as the martingale coupling between $\mu_h$ and itself,
\[
G_{h}(A)\geq \int_{\R^d\times \R^d} \frac12\langle A(z-y),z-y\rangle\,\pi^*_{h,\mu_h}(\d y,\d z)-\alpha_h(\mu_h)=0.
\]
Moreover, one readily verifies that $G_h\colon \sym\to \R$ is convex and nondecreasing with $G_h(0)=0$ for all $h\in (0,h_0)$.

Now, let $L\geq0$, $\ep>0$, and $a\geq0$ as in Lemma \ref{lem.apriori2}.\ Then, by Assumption \eqref{cond.Pprime}, there exists some $r>0$ such that $c(r)>0$ and $\frac{4Lar}{c(r)}\leq \ep$. Now, let $h\in (0,h_0)$ and $A_1,A_2\in \sym$ with $|A_i|\leq L$ for $i=1,2$ and $2r|A_1-A_2|\leq \ep$.\ Then, 
\begin{align*}
\bigg|\frac{G_{h}(A_1)}h-\frac{G_{h}(A_2)}h\bigg|&\leq \sup_{\nu \in \Pc_h^a}\int_{\R^{d}\times \R^{d}} |A_1-A_2|\frac{|z-y|^2}{2h}\,\pi_{h,\nu}^*(\d y,\d z)\\
&\leq {r|A_1-A_2|}+2L \int_{\big\{\frac{|z-y|^2}{2h}>r\big\}} \frac{|z-y|^2}{2h}\,\pi_{h,\nu}^*(\d y,\d z)\\
& \leq \frac{\ep}2+ \frac{2Lr}{c(r)} \int_{\R^{d}\times \R^{d}}c\bigg(\frac{|z-y|^2}{2h}\bigg)\,\pi_{h,\nu}^*(\d y,\d z)\leq \ep.
\end{align*}
The statement now follows from the Arzel\`a-Ascoli theorem together with a diagonal argument and the observation that convexity and monotonicity carry over from $G_h$ to the limiting functional $G$.
\end{proof}

The following theorem is the main result of this section.\ We denote by $ \Cb^{2}(\R^{d}) $ the space of all $ f \in \Cb^1(\R^{d})$ with bounded and continuous second derivative $\nabla^2 f\colon \R^d\to \sym$.	

\begin{theorem}\label{thm.generator2}\
Assume that conditions \eqref{cond.A} and \eqref{cond.Pprime} are satisfied.
\begin{enumerate}
\item[a)]
For all $R\geq 0$ and $f\in \Cb^2(\R^d)$,
\[
    \lim_{h\downarrow 0}\sup_{|x|\leq R}\bigg|\frac{(I_{h}f)(x)-(P_{h}f)(x)-G_{h}\big(\nabla^2 f(x)\big)}{h}\bigg|=0.
\]
 \item[b)] Let $(h_n)_{n\in \N}\subset (0,\infty)$  be a null sequence and $G\colon \sym\to \R$ such that \eqref{eq.existence.G} holds.\ Then, for all $ f \in \Cb^2(\R^{d})$,
 \begin{equation}\label{eq.thmgenerator2}
	\frac{I_{h_n}f - P_{h_n}f}{h_n} \to G(\nabla^2 f)\quad\text{as }n\to\infty.
 \end{equation}
\end{enumerate}
\end{theorem}

\begin{proof}
Let $f \in \Cb^2(\R^d)$. Then, using \eqref{eq.apriori.mart.lipschitz.thm} together with the fact that $ \|f\|_{\Lip,1}\leq \|\nabla^2 f\|_\infty$,
\[
\sup_{h\in (0,h_0)}\frac{\|I_hf-P_hf\|_\infty}{h}\leq c^*\big(\|\nabla^2 f\|_\infty\big),
\]
so that the claim in part b) again follows once we have established part a).\ To that end, let $a\geq 0$ as in Lemma \ref{lem.apriori2} with $L:=\|\nabla^2 f\|_\infty$, so that, for all $h\in (0,h_0)$ and $x\in \R^d$,
\begin{align*}
    \bigg|\frac{(I_hf)(x)-(P_hf)(x)}{h}-\frac{G_h\big(\nabla^2 f(x)\big)}{h}\bigg| &\leq  \sup_{\nu\in \Pc_h^a} \frac1h \bigg| \int_{\R^d}f\big(\psi_{h}(x)+z\big)\, \nu(\d z)-(P_hf)(x)\\
    &\quad - \int_{\R^d\times \R^d}\frac12\langle\nabla^2 f(x)(z-y),z-y\rangle \,\pi^*_{h,\nu}(\d y,\d z)\bigg|.
\end{align*}
Using Taylor's theorem together with the martingale condition \eqref{eq.def.mart.coupling}, for all $\nu\in \Pc_h^a$,
\begin{align*}
    \frac1h \bigg|\int_{\R^d}&f\big(\psi_{h}(x)+z\big)\, \nu(\d z)-(P_hf)(x)
    - \int_{\R^d\times \R^d}\frac12\langle\nabla^2 f(x)(z-y),z-y\rangle \,\pi^*_{h,\nu}(\d y,\d z)\bigg|\\
    &\quad\leq \int_{\R^d\times \R^d}\int_0^1s\Big|\nabla^2 f\big(\psi_h(x)+y+s(z-y)\big)-\nabla^2 f(x)\Big|\frac{|z-y|^2}{h} \,\d s\, \pi_{h,\nu}^*(\d y,\d z).
\end{align*}	
We split the last integral in three parts and estimate each one of them separately.\ Let $R\geq 0$ and $\ep>0$. Then, there exists $r>0$ such that $\frac{6Lar}{c(r)}\leq \ep$.\ Moreover, by \eqref{eq.ass.psi_h}, there exist $\overline h\in (0,h_0]$ and $\de>0$ such that
\begin{equation}\label{eq.unifcont.hessian}
    \sup_{|x|\leq R}\big|\nabla^2 f\big(\psi_h(x)+u\big)-\nabla^2 f(x)\big|\leq \frac{\epsilon}{3r}\quad\text{for all }h\in \big(0,\overline h\big)\text{ and }u\in \R^d\text{ with }|u|\leq 2\delta.
\end{equation} 
Last but not least, since $\mu_h\to \de_0$ as $h\downarrow 0$ by Assumption \eqref{cond.A}, after potentially modifying $\overline h\in (0,h_0]$, we may w.l.o.g.\ assume that $\sqrt{2r\overline h}\leq \de$ and
\begin{equation}\label{eq.cond.M.proof.main2}
6Lr\mu_h\big(\big\{y\in \R^d\,\big|\, |y|>\de\big\}\big)\leq\ep\quad \text{for all }h\in \big(0,\overline h\big).
\end{equation}
Then, for all $h\in \big(0,\overline h\big)$, $\nu\in \Pc_h^a$, and $x\in \R^d$ with $|x|\leq R$, by \eqref{eq.unifcont.hessian},
\begin{align*}
    \int_{\{|y|\leq \delta\}\cap \big\{\frac{|z-y|^2}{2h}\leq r\big\}}\int_0^1 s\Big|\nabla^2 f\big(\psi_h(x)+y+s(z-y)\big)-\nabla^2 f(x)\Big|&\frac{|z-y|^2}{h}\,\d s\,\pi_{h,\nu}^*(\d y,\d z)\leq \frac{\ep}{3},
\end{align*}
by \eqref{eq.cond.M.proof.main2},
\begin{align*}
    \int_{\{|y|> \delta\}\cap \big\{\frac{|z-y|^2}{2h}\leq r\big\}}&\int_0^1 s\Big|\nabla^2 f\big(\psi_h(x)+y+s(z-y)\big)-\nabla^2 f(x)\Big|\frac{|z-y|^2}{h}\,\d s\,\pi_{h,\nu}^*(\d y,\d z)\\
    &\leq 2Lr\mu_h\big(\big\{y\in \R^d\,\big|\, |y|>\de\big\}\big)\leq\frac{\ep}3,
\end{align*}
and
\begin{align*}
    &\int_{\big\{\frac{|z-y|^2}{2h}> r\big\}}\int_0^1 s\Big|\nabla^2 f\big(\psi_h(x)+y+s(z-y)\big)-\nabla^2 f(x)\Big|\frac{|z-y|^2}{h}\,\d s\,\pi_{h,\nu}^*(\d y,\d z)\\
    &\qquad\leq 2L \int_{\big\{\frac{|z-y|^2}{2h}> r\big\}} \frac{|z-y|^2}{2h}\,\pi_{h,\nu}^*(\d y,\d z)\leq \frac{2Lr}{c(r)}\int_{\R^d\times \R^d} c\bigg(\frac{|z-y|^2}{2h}\bigg)\,\pi_{h,\nu}^*(\d y,\d z)\\
    & \qquad\leq \frac{2Lar}{c(r)}\leq\frac{\ep}3.
\end{align*}
The proof is complete.
\end{proof}

Again, we provide some examples for penalizations that satisfy Assumption \eqref{cond.Pprime}.

\begin{example}[Martingale Wasserstein penalization]\label{ex.wassersteinsop}
Let $\phi\colon [0,\infty)\to [0,\infty]$ convex and lower semicontinuous with $\phi(0)=0 $ and $\phi(v)\neq 0$ for some $v\in (0,\infty)$.\ Again, $\phi$ is nondecreasing and continuous on $ \operatorname{dom}(\phi):= \{v \in [0,\infty)\,|\, \phi(v) < \infty\} $. Further, let $ p \in (2,\infty)$, and assume that the map $v \mapsto \phi(v^{2/p})$ is convex. Since $\phi \not \equiv 0$, this implies that $\liminf_{v \to \infty} \frac{\phi(v)}{v^{p/2}}>0$. We consider the penalization
\[  
    \alpha_t(\nu):=t\phi\bigg(\frac{\Wc_p^{\mart}(\mu_t,\nu)}{2t}\bigg) \quad\text{for all }t>0\text{ and }\nu\in \Pc(\R^d),
\]
where $\Wc_p^{\mart}(\mu_t,\nu)$ denotes the squared $p$-martingale Wasserstein distance between $\mu_t$ and $\nu\in \Pc(\R^d)$, i.e.,
\[
    \Wc_p^{\mart}(\mu_t,\nu):=\bigg(\inf_{\pi\in \mart(\mu_t,\nu)}\int_{\R^d\times \R^d} |z-y|^p\,\pi(\d y,\d z)\bigg)^{2/p},
\]
 and we set $\phi(\infty):=\infty$.\ Again, since $ \phi(0)= 0$, it follows that $\alpha_t(\mu_t)=0$ for all $t>0$. Since $\liminf_{v \to \infty} \frac{\phi(v)}{v^{p/2}}>0$ and $\phi$ is nondecreasing, for all $ \epsilon >0 $, there exists some $a \geq 0$ such that
\[
    \phi(v) \geq \epsilon v^{p/2}, \quad \text{for all }v \geq a.
\]
This yields that for all $ v \geq 0 $ we have 
\[
    \phi(v) = \phi(v)- \epsilon v^{p/2} + \epsilon v^{p/2} \geq \left(\min_{u \in [0,a]}\phi(u) - \epsilon u^{p/2}\right) + \epsilon v^{p/2} = \epsilon v^{p/2} - c, 
\]
where $ c := - \min_{u \in [0,a]}(\phi(u) - \epsilon u^{p/2}) \geq 0 $. Let $ t > 0 $ and $ \nu \in \Pc_t$.\ Then, 
\begin{align*}
	\frac{\alpha_t(\nu)}{t} &= \phi\left(\frac{\Wc_p^{\mart}(\mu_t,\nu)}{2t}\right)
	\geq \epsilon \bigg(\frac{\Wc_p^{\mart}(\mu_t,\nu)}{2t}\bigg)^{p/2}-c\\
	&= \inf_{\pi \in \mart(\mu_t,\nu)} \Bigg(\int_{\R^{d} \times \R^{d}} \epsilon\bigg(\frac{|y-z|^2}{2t}\bigg)^{p/2}-c \,\pi(\d y,\d z)\Bigg).
\end{align*}
Defining $ c(v):= \epsilon v^{p/2}-c$ for all $v\geq 0$, shows that Assumption \eqref{cond.Pprime} is satisfied.\ 

For $A\in \sym$, let
\[
\la (A):=\sup_{|x|=1} \langle Ax,x\rangle\in \R. 
\]
Let $t >0$ and $A \in \sym$. Then,
\begin{align*}
	\frac{G_t(A)}{t}&
	= \sup_{\nu \in \Pc_t}\Bigg(\inf_{\pi\in \mart(\mu_t,\nu)} \int_{\R^d\times \R^d}\frac{1}{2t}\big\langle A(z-y),z-y\big\rangle\,\pi(\d y, \d z) - \phi\bigg(\frac{\Wc_p^{\mart}(\mu_t,\nu)}{2t}\bigg)\Bigg)\\
    &\leq  \phi^*\big(\la(A)\big).
\end{align*}
On the other hand, since the map $[0,\infty)\to [0,\infty],\; v \mapsto \phi(v^{2/p})$ is convex, Jensen's inequality implies that
\begin{equation}\label{ineq.wasserstein2} 
	\phi\bigg(\frac{\Wc_p^{\mart}(\mu_t,\nu)}{2t}\bigg)\leq \inf_{\pi \in \mart(\mu_t,\nu)} \int_{\R^{d} \times \R^{d}} \phi\bigg(\frac{|y-z|^2}{2t}\bigg)\, \pi(\d y,\d z).
\end{equation}
This yields that, for all $A\in \sym$,
\begin{align*}
	\frac{G_t(A)}{t}\geq \sup_{\theta\in \R^d}\Bigg(\frac{1}{2t}&\langle A\theta ,\theta \rangle- \phi\bigg(\frac{|\theta|^2}{2t}\bigg)\Bigg)= \phi^*\big(\la(A)\big),
	\end{align*}
 where the first inequality is achieved by considering $\nu=\mu_t\ast {\rm B}_{s}^\theta\in \Pc_t$ with $${\rm B}_{s}^\theta\big(\{\theta\}\big)={\rm B}_{s}^\theta\big(\{-\theta\}\big)=\frac12$$ for all $\theta\in \R^d$. This shows that $\frac{G_t(A)}t=\phi^*\big(\la(A)\big)$ for all $t>0$ and $A\in \sym$.
\end{example}

\begin{example}[Martingale optimal transport penalization]\label{ex.sop}\
Let $\phi\colon \R^d\to [0,\infty)$ with $\phi(0)=0$ and $\lim_{|u|\to \infty}\frac{\phi(u)}{|u|^2}= \infty$.\ For all $t>0$, and $\nu\in \Pc(\R^d)$, let
    \[
        \alpha_t(\nu):=\inf_{\pi\in \mart(\mu_t,\nu)} \int_{\R^d\times \R^d} t\phi\bigg(\frac{z-y}{\sqrt{2t}}\bigg)\, \pi(\d y,\d z)\in [0,\infty].
    \]
    We proceed in a similar way as in Example \ref{ex.fop} to show that Assumption \eqref{cond.Pprime} is satisfied, and define
    \[
        c(v):=\inf_{|u|^2\geq v} \phi(u)\quad\text{for all }v\geq 0.
    \]
    Then, by assumption, $c\colon [0,\infty)\to [0,\infty)$ is nondecreasing with
    \[
        \frac{c(v)}{v}\geq \inf_{|u|^2\geq v}\frac{\phi(u)}{|u|^2}\to \infty\quad \text{as }v\to \infty.
    \]
    Moreover, $\phi(u)\geq c(|u|^2)$ for all $u\in \R^d$, so that
    \begin{align*}
        \alpha_t(\nu)&= \inf_{\pi\in \mart(\mu_t,\nu)} \int_{\R^d\times \R^d} t\phi\bigg(\frac{z-y}{\sqrt{2t}}\bigg)\, \pi(\d y,\d z)\\
        &\geq \inf_{\pi\in \mart(\mu_t,\nu)} \int_{\R^d\times \R^d} tc\bigg(\frac{|z-y|^2}{2t}\bigg)\, \pi(\d y,\d z)
    \end{align*}
    for all $t>0$ and $\nu\in \Pc_t$.\ For $A\in \sym$, we define
    \begin{equation}
        \overline \ph(A):=\sup_{x\in \R^d} \Big(\langle Ax,x\rangle-\ph(x)\Big).
    \end{equation}
    Since $\ph(0)=0$ and $\lim_{|u|\to \infty}\frac{\phi(u)}{|u|^2}= \infty$, it follows that $\overline \ph(A)\in [0,\infty)$ for all $A\in \sym$.
    
    Now, let $A\in \sym$, $t>0$, and $\nu\in \Pc_t$.\ Then, for any $\pi^*\in \mart(\mu_t,\nu)$ with \eqref{eq.almost.mart.opt.coupling}, it follows that
    \begin{align*}
        \int_{\R^d\times \R^d} &\frac{1}{2t}\big\langle A(z-y), z-y \big\rangle\, \pi^*(\d y,\d z)-\frac{\alpha_t(\nu)}{t}\\
        &\qquad\qquad\leq \int_{\R^d\times \R^d} \frac{1}{2t}\big\langle A(z-y), z-y\big\rangle-\phi\bigg(\frac{z-y}{\sqrt{2t}}\bigg)\, \pi^*(\d y,\d z)\leq \overline\ph(A).
    \end{align*}
    Taking the supremum over all $\nu\in \Pc_t$ yields that $\frac{G_t(A)}{t}\leq \overline\phi(A)$.\ On the other hand, proceeding as in Example \ref{ex.wassersteinsop}, for all $t>0$ and $A\in \sym$, it follows that
    \[
        \frac{G_t(A)}{t}\geq \sup_{\theta\in \R^d} \Bigg(\frac1{2t}\langle A\theta,\theta\rangle -\phi\bigg(\frac{\theta}{\sqrt{2t}}\bigg)\Bigg)=\overline\phi(A),
    \]
     which shows that $\frac{G_t(A)}{t}=\overline\phi(A)$.
    \end{example}

    \begin{example}[Volatility control penalization]\label{ex.vola.control}
    Let $\phi\colon [0,\infty)\to [0,\infty]$ be convex and lower semicontinuous with $\phi(0)=0$ and $\phi(v)\neq 0$ for some $v\in (0,\infty)$.\ Moreover, let $ p \in (2,\infty)$, and assume that the map $v \mapsto \phi(v^{2/p})$ is convex.\ Since $\phi \not \equiv 0$, this implies that $\liminf_{v \to \infty} \frac{\phi(v)}{v^{p/2}}>0$.
    Further, let $(\Omega,\mathcal F,\P)$ be a complete probability space and $(Y_t)_{t> 0}$ be a family of random variables on $(\Omega,\mathcal F,\P)$ with $Y_t\sim \mu_t$ for all $t>0$.\ Let $\mathbb F=(\mathcal F_t)_{t\geq 0}$ be a filtration on $(\Omega,\mathcal F)$ satisfying the usual conditions such that $Y_t$ is independent of $\mathcal F_t$ for all $t>0$, and $(W_t)_{t\geq0}$ be an $\mathbb F$-Brownian motion.\ Denote by $\mathcal A$ the set of all $\mathbb F$-progressively measurable $\psd$-valued stochastic processes $\sigma=(\sigma_t)_{t\geq 0}$ with $\E\big(\int_0^t \phi(\tfrac12|\sigma_s|^2)\, \d s\big)<\infty$ for all $t\geq 0$, where $\psd$ denotes the set of all positive semidefinite elements of $\sym$.\ For $t>0$ and $\sigma\in \mathcal A$, let $\mu_t^\sigma\in \Pc(\R^d)$ denote the law of $Y_t+\int_0^t\sigma_s\, \d W_s$ and, for $\si\in \psd$, let
    \[
    |\si|_{{\rm HS}}:=\sqrt{\trace(\si^2)}
    \]
 denote the Hilbert-Schmidt norm of $\si$.\ 
    For $t>0$ and $\nu\in \Pc(\R^d)$, we define
    \[
        \alpha_t(\nu):=\inf \bigg\{ \E\bigg[\int_0^t \phi\bigg(\frac{|\sigma_s|_{{\rm HS}}^2}{2}\bigg)\, \d s\bigg]\,\bigg|\, \sigma\in \mathcal A,\, \mu_t^\sigma=\nu\bigg\},
    \]
    where we use the convention $\inf\emptyset=\infty$.\ Then, using Jensen's inequality and the {Burkholder-Davis-Gundy inequality \cite[Theorem 3.28]{KaratzasShreve1991}},
        \begin{align*}
        \alpha_t(\nu)&=\inf\bigg\{ \E\bigg[ \int_0^t \phi\bigg(\frac{|\sigma_s|_{{\rm HS}}^2}2\bigg)\, \d s\bigg]\,\bigg|\, \sigma\in \mathcal A,\, \mu_t^\sigma=\nu\bigg\}\\
        &\geq \inf\Bigg\{ t\phi\Bigg(\E\bigg[ \bigg(\frac1{2t}\int_0^t |\sigma_s|_{{\rm HS}}^2\, \d s\bigg)^{p/2}\bigg]^{2/p}\Bigg)\,\bigg|\, \sigma\in \mathcal A,\, \mu_t^\sigma=\nu\Bigg\}\\
        &\geq \inf\Bigg\{ t\phi\Bigg(\frac{C_p^{-2/p}}{2t}\E\bigg(\bigg|\int_0^t \sigma_s\, \d W_s\bigg|^p\bigg)^{2/p}\Bigg)\,\bigg|\, \sigma\in \mathcal A,\, \mu_t^\sigma=\nu\Bigg\}\\
        &\geq t\phi\bigg(\frac{C_p^{-2/p}\Wc_p^{\mart}(\mu_t,\nu)^2}{2t}\bigg)\quad\text{for all }t>0\text{ and }\nu\in \Pc_t,
    \end{align*}
    where $C_p>0$ is the constant from the Burkholder-Davis-Gundy inequality.\ Hence, by Example \ref{ex.wassersteinsop} with $\phi\big(C_p^{-2/p} \cdot\big)$, it follows that Assumption \eqref{cond.Pprime} is satisfied.\ Moreover, for all $t>0$ and $A\in \sym$, choosing $\si_s\equiv\si\in \psd$ for all $s\in [0,t]$, it follows that
    \[
        \frac{G_t(A)}{t}\geq \sup_{\si\in \psd} \Big(\trace\big(\si A\si\big) -\phi\big({|\si|_{{\rm HS}}^2}\big)\Big)
    \]
    and, by It\^o's isometry,
    \begin{align*}
        \frac{G_t(A)}{t}\leq  \sup_{\sigma\in \mathcal A} \E\bigg[\frac1{2t}\int_0^t \trace\big(\si_s A\si_s\big) -\phi\bigg(\frac{|\sigma_s|_{{\rm HS}}^2}2\bigg)\d s\bigg]\leq \sup_{\si\in \psd} \Big(\trace\big(\si A\si\big) -\phi\big({|\si|_{{\rm HS}}^2}\big)\Big).
    \end{align*}
    Hence, $$\frac{G_t(A)}{t}=\sup_{\si\in \psd} \Big(\trace\big(\si A\si\big) -\phi\big({|\si|_{{\rm HS}}^2}\big)\Big)=\phi^*\big(\la(A)\big)$$ for all $t>0$ and $A\in \sym$.
\end{example}

\section{Chernoff approximation and control representation}\label{sec.chernoff}

We now pass from the local expansions of the robust one-step operators to dynamically consistent risk measures in continuous time through a suitable iteration.\ More precisely, this section provides Chernoff approximations for the first- and second-order perturbations, and identifies the corresponding limiting semigroups through stochastic control representations.

Throughout this section we set $I_0:=\id_{\Cb(\R^d)}$ and use the notion of convergence in
the mixed topology introduced in Section~\ref{sec.setup}.  If $(h_n)_{n\in\N}\subset(0,\infty)$
is a null sequence, $t\geq0$, and $(k_n)_{n\in\N}\subset\N$ satisfies $k_nh_n\to t$ as $n\to \infty$, we write
\[
        I_{h_n}^{k_n}f:=\underbrace{I_{h_n}\circ\cdots\circ I_{h_n}}_{k_n\text{ times}}f .
\]
A typical choice is $h_n=t/n$ and $k_n=n$, which yields the approximation discussed in the introduction.

We adopt the following standing assumptions on the reference dynamics from \cite{nendel2025chernoff}, which imply condition \eqref{cond.A} from Section \ref{sec.setup}:
\begin{enumerate}
    \item[(M)]\makeatletter\def\@currentlabel{M}\makeatother\label{cond.M} There exists $h_0>0$ such that
	   \begin{equation}
           \notag
	       \sup_{h\in (0,h_0)}\frac1h\Bigg( \int_{\R^d} {1\wedge |y|^2}\,\mu_h(\d y)+\bigg|\int_{\{|y|\leq 1\}} y\,\mu_h(\d y)\bigg|\Bigg)<\infty. 
	   \end{equation}
       \item[(T)]\makeatletter\def\@currentlabel{T}\makeatother\label{cond.T} For all $\ep>0$, there exists $M_\ep>0 $ such that
	   \begin{equation}\label{eq.Ass.M2}
	       \limsup_{h\downarrow0}\frac{\mu_h\big(\big\{y\in \R^d\,\big|\, |y|>M_\ep\big\}\big)}h< \ep.
	   \end{equation}
        \item[(D)]\makeatletter\def\@currentlabel{D}\makeatother\label{cond.D} There exist $\omega\geq 0$, $\delta>0$, and $h_0>0$ such that, for all $u\in \R^d$ with $|u|\leq \delta$,
	\begin{equation}\label{eq.biton}
		\sup_{h\in (0,h_0)}\sup_{x\in \R^d}\frac{|\psi_h(x+u)-\psi_h(x)-u|}{h}\leq \omega |u|.
	\end{equation}   
    Moreover, $\limsup_{h\downarrow0}\frac{|\psi_h(0)|}{h}<\infty$. 
\end{enumerate} 
We point out that there is no relation between the conditions \eqref{cond.M} and \eqref{cond.T} on the one side and condition \eqref{cond.D} on the other side, so that the choice of the family $(\psi_t)_{t>0}$ is still completely independent of the family $(\mu_t)_{t>0}$, see \cite{nendel2025chernoff} for a detailed discussion of these conditions.

In \cite[Theorem 2.1]{nendel2025chernoff}, it is shown that, for every null sequence in $(0,\infty)$, there exist a subsequence $(h_n)_{n\in \N}$, a L\'evy process $Y=(Y_t)_{t\geq0}$ on some probability space $(\Omega,\mathcal F,\P)$, and an $\omega$-Lipschitz function $F\colon \R^d\to \R^d$ such that, for all $t\geq 0$, $f\in \Cb(\R^d)$, $(k_n)_{n\in \N}\subset \N$ with $k_nh_n\to t$ as $n\to \infty$, $(f_n)_{n\in \N}\subset \Cb(\R^d)$ with $f_n\to f$ in the mixed topology as $n\to \infty$, and $r\geq 0$,
	\[
	\sup_{|x|\leq r}\Big|\big(P_{h_n}^{k_n}f_n\big)(x)- \E_\P[f(X_t^x)]\Big|\to 0\quad \text{as }n\to \infty,
	\]
	where, for all $x\in \R^d$, the process $(X_t^x)_{t\geq 0}$ is the unique strong solution to the L\'evy SDE
	\begin{equation}\label{eq.reference.levy.sde}
	\d X_t^x= F(X_t^x)\,\d t+\d Y_t\quad \text{with}\quad X_0^x=x.
	\end{equation}
	Moreover, if $f\in \Ccinf(\R^d)$, then
    \begin{equation}\label{eq.paper.max}
    \lim_{h\downarrow0}\frac{\E_\P[f(X_{h}^\cdot)]-f}{h}=\lim_{n\to \infty}\frac{P_{h_n}f-f}{h_n}=\langle \nabla f,F\rangle+L_{(b,\Sigma,\nu)}f\quad \text{in the mixed topology,}
    \end{equation}
    where $L_{(b,\Sigma,\nu)}$ is the L\'evy generator associated with the L\'evy triplet $(b,\Sigma,\nu)$ of $Y$.\ 

    Throughout the remainder of this section, the sequence $(h_n)_{n\in \N}$, the probability space $(\Om,\mathcal F,\P)$, which we assume to be complete, the L\'evy process $Y$ with L\'evy triplet $(b,\Sigma,\nu)$, and the Lipschitz function $F$ shall be fixed.

All limiting statements below are therefore formulated along this fixed subsequence.  The remaining task is to identify the additional nonlinear contribution generated by the first- and second-order transport costs.

\subsection{First-order optimal control problem}\label{subsec.first.chernoff}

In this section, we consider the first-order transport regime of Section~\ref{sec:first order} assuming Condition~\eqref{cond.P}. By potentially passing to a further subsequence along the fixed null sequence $(h_n)_{n\in \N}$, we may assume that there exists a finite convex function $g\colon\R^d\to\R$ such that
\begin{equation}\label{eq.ass.g.convergence.section5}
        \sup_{|m|\leq L}\left|\frac{g_{h_n}(m)}{h_n}-g(m)\right|\longrightarrow0
        \quad\text{for all }L\geq0.
\end{equation}
    Let $\phi\colon \R^d\to [0,\infty]$  with $\phi(0)=0$, $\lim_{|u|\to \infty}\frac{\phi(u)}{|u|}=\infty$, and $\phi^*=g$.\footnote{An explicit example is given by the Fenchel-Legendre transform $\varphi(x):=\sup_{m\in \R^d}(\langle m,x\rangle-g(m))$ for $x\in \R^d$.}\ Moreover, let $\mathcal A$ denote the set of all progressively measurable controls w.r.t.\ the filtration generated by $Y$ augmented by all $\P$-null sets $\be=(\be_t)_{t\geq 0}\subset \R^d$  with
\[
\E\bigg[\int_0^t|\be_s|\,\d s\bigg]<\infty.
\]
For $\be\in \mathcal A$, we then consider controlled dynamics of the form
\begin{equation}\label{eq.condtolled.levy.sde1}
        \d X_t^{x,\be}=\big(\be_t+F(X_t^{x,\be})\big)\,\d t+\d Y_t,
        \qquad X_0^{x,\be}=x
\end{equation}
and the value function
\[
(V_tf)(x):=\sup_{\be\in \mathcal A} \E\bigg[f\big(X_t^{x,\be}\big)-\int_0^t\ph(\be_s)\, \d s\bigg]
\]
for $t\geq 0$, $f\in \Cb(\R^d)$, and $x\in \R^d$.\ The family $(V_t)_{t\geq 0}$ forms the continuous-time counterpart of the first-order robust one-step operators from Section \ref{sec:first order}.\ The control acts on the drift, and the running cost is chosen so that the Hamiltonian is $g(\nabla f)$.

\begin{remark}\label{rem.first.chernoff.pde}
 The family $(V_t)_{t\geq 0}$ is a strongly continuous convex monotone semigroup on $\Cb(\R^d)$ in the sense of \cite[Definition 3.1]{BlessingDenkKupperNendel2025}.\ One readily verifies that $V_t\colon \Cb(\R^d)\to \Cb(\R^d)$ is convex and monotone for all $t\geq 0$ and $V_0f=f$ for all $f\in \Cb(\R^d)$.\ Moreover, the semigroup property $V_tV_s=V_{t+s}$ corresponds to the dynamic programming principle.\ The strong continuity follows from the a priori estimate
 \begin{align}
\notag \big|V_tf(x)-\E\big[f\big(X_t^{x,0}\big)\big]\big| & \leq \sup_{\be\in \mathcal A} \E\bigg[\int_0^te^{\om t}\|f\|_{\Lip}|\be_s|-\phi(\be_s)\, \d s\bigg]\\
&\leq t\sup_{b\in \R^d} \big(e^{\om t}\|f\|_{\Lip}|b|-\varphi(b)\big) \label{apriori.SG.first} 
\end{align}
for $t\geq0$, $f\in \Lipb(\R^d)$, and $x\in \R^d$ together with \cite[Proposition C.3]{nendel2025chernoff} and the fact that $\Lipb(\R^d)$ is dense in $\Cb(\R^d)$ in the mixed topology.\ Moreover, using \eqref{apriori.SG.first} together with \cite[Equation (C.2)]{nendel2025chernoff}, for all $\ep>0$, there exist $\delta,t_0>0$ such that
\[
V_t \big(f(x+\cdot\,)\big)\leq (V_tf)(x+\cdot\, )+\ep t
\]
for all $t\in [0,t_0]$, $f\in \Lipb(\R^d)$, and $x\in \R^d$ with $|x|\leq \de$.\ Moreover, using the fundamental theorem of calculus, one readily verifies that, for all $K\Subset\R^d$,
\begin{equation}\label{eq.gen.control.first}
\lim_{h\downarrow 0} \sup_{x\in K}\bigg|\frac{V_hf(x)-\E\big[f\big(X_h^{x,0}\big)\big]}{h}-g\big(\nabla f(x)\big)\bigg|=0\quad \text{for all }f\in \Cc^1(\R^d).
\end{equation}
\end{remark}

The first-order correction is therefore separated from the reference L\'evy dynamics at the level of the generator.\ The next subsection gives the analogous representation when the local perturbation is of second order.

\subsection{Second-order optimal control problem}\label{subsec.second.chernoff}

We now turn to the martingale optimal  transport setup from Section \ref{sec:second order}.\ By potentially passing to a further subsequence along the null sequence $(h_n)_{n\in \N}$, we may assume that there exists a finite convex function
$G\colon\sym\to\R$ such that
\begin{equation}\label{eq.ass.G.convergence.section5}
        \sup_{|A|\leq L}\left|\frac{G_{h_n}(A)}{h_n}-G(A)\right|\longrightarrow0
        \quad\text{for all }L\geq0.
\end{equation}
After potentially enriching the probability space $(\Om,\mathcal F,\P)$, we may w.l.o.g.\ assume that there exists  $d$-dimensional Brownian motion $W=(W_t)_{t\geq0}$ that is independent of $Y$.\ Now, let $\phi\colon \sym_+\to [0,\infty]$ with $\phi(0)=0$, $\lim_{|\si|_{\rm HS}\to \infty}\frac{\phi(\si)}{|\si|_{\rm HS}^2}=\infty$, and 
\[
\sup_{\si\in \sym_+} \Big(\trace(\si A\si)-\phi(\si)\Big)=G(A)\quad \text{for all }A\in \sym.\footnote{An explicit example is given by the Fenchel-Legendre transform $\varphi(\si):=\sup_{A\in \sym}(\trace(\si A\si)-G(A))$ for $\si\in \psd$.}
\]
 In this section, we consider the set $\mathcal A$ of all progressively measurable controls w.r.t.\ the filtration generated by $(Y,W)$ augmented by all $\P$-null sets $\si=(\si_t)_{t\geq 0}\subset \sym_+$ with
\[
\E\bigg[\int_0^t |\si_s|_{\rm HS}^2\, \d s\bigg]<\infty.
\]
For $\si\in \mathcal A$, we then consider controlled dynamics of the form
\begin{equation}\label{eq.condtolled.levy.sde2}
        \d X_t^{x,\si}=F(X_t^{x,\si})\,\d t+\d Y_t+ \si_t\, \d W_t,
        \qquad X_0^{x,\si}=x
\end{equation}
and the value function
\[
(V_tf)(x):=\sup_{\si\in \mathcal A} \E\bigg[f\big(X_t^{x,\si}\big)-\int_0^t\phi\bigg(\frac{\si_s}{\sqrt{2}}\bigg)\, \d s\bigg]
\]
for $t\geq 0$, $f\in \Cb(\R^d)$, and $x\in \R^d$.\ Here, the control changes the quadratic variation of the additional Brownian component.\ Accordingly, the local nonlinear term acts on the Hessian and is described by $G(\nabla^2 f)$.

\begin{remark}\label{rem.second.chernoff.pde}
 Again, the family $(V_t)_{t\geq 0}$ is a strongly continuous convex monotone semigroup on $\Cb(\R^d)$ in the sense of \cite[Definition 3.1]{BlessingDenkKupperNendel2025}.\ As in Section \ref{subsec.first.chernoff}, $V_t\colon \Cb(\R^d)\to \Cb(\R^d)$ is convex and monotone for all $t\geq 0$, $V_0f=f$ for all $f\in \Cb(\R^d)$, and the semigroup property $V_tV_s=V_{t+s}$ corresponds to the dynamic programming principle.\ Now, the strong continuity follows from the estimate 
 \begin{align}
\notag \big|V_tf(x)-\E\big[f\big(X_t^{x,0}\big)\big]\big|  &\leq \sup_{\si\in \mathcal A} \E\bigg[\int_0^te^{2\om^2 t^2}\|\nabla^2f\|_{\infty}|\si_s|_{\rm HS}^2-\phi\bigg(\frac{\si_s}{\sqrt{2}}\bigg)\, \d s\bigg]\\
&\leq t\sup_{\si\in \sym_+} \big(2e^{2\om^2 t^2}\|\nabla^2 f\|_{\infty}|\si|_{\rm HS}^2-\varphi(\si)\big) \label{apriori.SG.second}
\end{align}
for $t\geq0$, $f\in \Cb^2(\R^d)$, and $x\in \R^d$ together with \cite[Proposition C.3]{nendel2025chernoff} and the fact that $\Cb^2(\R^d)$ is dense in $\Cb(\R^d)$ in the mixed topology.\ Moreover, using Gronwall's lemma,
\[
|x+X_t^{y,\si}-X_t^{x+y,\si}|\leq \om te^{\om t}|x|\quad \P\text{-a.s.\ for all }x,y\in \R^d \text{ and }t\geq 0.
\]
Using this estimate, for all $\ep>0$, there exist $\delta,t_0>0$ such that
\[
V_t \big(f(x+\cdot\,)\big)\leq (V_tf)(x+\cdot\, )+\ep t
\]
for all $t\in [0,t_0]$, $f\in \Lipb(\R^d)$, and $x\in \R^d$ with $|x|\leq \de$.\ Using It\^o's formula, one readily verifies that, for all $K\Subset \R^d$,
\begin{equation}\label{eq.gen.control.second}
\lim_{h\downarrow 0} \sup_{x\in K}\bigg|\frac{V_hf(x)-\E\big[f\big(X_h^{x,0}\big)\big]}{h}-G\big(\nabla^2 f(x)\big)\bigg|=0\quad \text{for all }f\in \Cc^2(\R^d).
\end{equation}
\end{remark}

\subsection{Identification as stochastic control problems}\label{subsec.control.identification}

We now identify the control semigroups constructed in the previous subsections as the limits of the iterated robust one-step schemes.\ The reference dynamics has already been fixed through \eqref{eq.paper.max}, while the first- and second-order perturbations are identified by \eqref{eq.gen.control.first} and \eqref{eq.gen.control.second}.\ The comparison principle for convex monotone semigroups from \cite{BlessingDenkKupperNendel2025} then turns the agreement of generators on smooth test functions into convergence of the Chernoff products.

\begin{theorem}[Chernoff approximation]\label{thm.chernoff.section5}
Consider the setup described in Subsection \ref{subsec.first.chernoff} or Subsection \ref{subsec.second.chernoff}.\ Then, for all $t\geq 0$, $f\in \Cb(\R^d)$, $(k_n)_{n\in \N}\subset \N$ with $k_nh_n\to t$ as $n\to \infty$, $(f_n)_{n\in \N}\subset \Cb(\R^d)$ with $f_n\to f$ in the mixed topology as $n\to \infty$, and $r\geq 0$,
	\[
	\sup_{|x|\leq r}\Big|\big(I_{h_n}^{k_n}f_n\big)(x)- (V_tf)(x)\Big|\to 0\quad \text{as }n\to \infty.
	\]
	Moreover, if $f\in \Cb(\R^d)$ such that $\big(\frac{I_{h_n}f-f}{h_n}\big)_{n\in \N}$ converges in the mixed topology, then
    \[
\lim_{h\downarrow0}\frac{V_hf-f}{h}=\lim_{n\to \infty}\frac{I_{h_n}f-f}{h_n}\quad \text{in the mixed topology.}
    \]
\end{theorem}

\begin{proof} 
One readily verifies that the family of operators $(I_t)_{t>0}$ from Section \ref{sec:first order} and Section \ref{sec:second order} satisfies the following properties:
\begin{enumerate}[label=\textnormal{(\roman*)}]
\item $I_t$ is convex and monotone for all $t\geq 0$
\item $I_t(f+c)=I_t f+c$ for all $f\in \Cb(\R^d)$ and $c\in \R$.
\item $\|I_hf\|_{\Lip}\leq e^{\om h}\|f\|_{\Lip}$ for all $h\in (0,h_0)$ and $f\in \Lipb(\R^d)$.
\end{enumerate}
The properties (i) and (ii) imply that
\[
        \|I_t f-I_t g\|_\infty\leq \|f-g\|_\infty\quad \text{for all }f,g\in \Cb(\R^d).
\]
Therefore, by Lemma \ref{lem.cont.above} together with \eqref{eq.apriori.lipschitz}, \eqref{eq.apriori.lipschitz.mart}, and \cite[Proposition 4.4 a) and Proposition 5.6 c)]{nendel2025chernoff}, the family $(I_{h_n})_{n\in \N}$ satisfies \cite[Assumption 2.10]{BlessingKupperNendel2025}.\ Using Theorem \ref{thm.generator}, Theorem \ref{thm.generator2}, Equation \eqref{eq.paper.max}, Equation \eqref{eq.gen.control.first}, and Equation \eqref{eq.gen.control.second},
\[
\lim_{n\to \infty} \frac{I_{h_n}f-f}{h_n}=\lim_{h\downarrow 0}\frac{V_{h}f-f}{h}\quad \text{in the mixed topology for all }f\in \Ccinf(\R^d).
\]
Using the fact that, by Remark \ref{rem.first.chernoff.pde} and Remark \ref{rem.second.chernoff.pde}, $(V_t)_{t\geq0}$ is a convex monotone semigroup together with Lemma \ref{lem.uniqueness}, \cite[Theorem 2.11]{BlessingKupperNendel2025}, and \cite[Theorem 4.9]{BlessingDenkKupperNendel2025}, the claim follows.
\end{proof}

\appendix

\section{Auxiliary results}\label{app.C}

The following lemma discusses some elementary properties of the Fenchel-Legendre transform.

\begin{lemma}\label{lem.conjugate}
 Let $c\colon [0,\infty)\to \R\cup\{\infty\}$ be nondecreasing with $c(0)\leq 0$ and $\frac{c(v)}{v}\to \infty$ as $v\to \infty$.
 \begin{enumerate}
    \item[a)] The map $c^*\colon \R\to [0,\infty)$, given by
    \begin{equation}\label{eq.def.conjugate}
     c^*(u):= \sup_{v\geq 0} \big(uv-c(v)\big)\quad\text{for all }u\in\R,
    \end{equation}
    is well-defined, convex, and nondecreasing with $c^*(u)=-c(0)$ for all $u\leq 0$. 
    \item[b)] Let $c(v)<\infty$ for all $v\geq0$.\ Then, there exists a convex nondecreasing function $c^{**}\colon [0,\infty)\to \R$ with $\frac{c^{**}(v)}{v}\to \infty$ as $v\to \infty$ and
    $$c(0)\leq c^{**}(v)\leq c(v) \quad\text{for all }v\geq 0.$$
 \end{enumerate}
\end{lemma}

\begin{proof}\
\begin{enumerate}
    \item[a)] Since $c$ is nondecreasing, $c(0)\leq 0$ and $\frac{c(v)}{v}\to \infty$ as $v\to \infty$, it follows that $0\leq c^*(u)<\infty$ for all $u> 0$.\ Moreover, for $u\leq 0$,
    \[
    -c(0)\leq c^*(u)\leq -\inf_{v\geq 0} c(v)=-c(0).
    \]
    By definition, $c^*$ is convex and nondecreasing. 
    \item[b)] Let $c(v)<\infty$ for all $v\geq0$ and $c^{**}\colon [0,\infty)\to \R$ be given by
    \[
    c^{**}(v):=\sup_{u\geq 0} \big(uv-c^*(u)\big)\quad\text{for all }v\geq 0.
    \]
    Then, $c^{**}$ is convex and nondecreasing. Moreover, $$c^{**}(0)=-\inf_{u\geq 0}c^*(u)=-c^*(0)=c(0)$$ and
    \[
    c^{**}(v)=\sup_{u\geq 0}\big(uv-c^*(u)\big)\leq \sup_{u\geq 0}\big(uv-(uv-c(v))\big)=c(v)\quad\text{for all }v\geq 0.
    \]
    Now, let $M\geq 0$. Since $c$ is nondecreasing with $\frac{c(v)}{v}\to \infty$ as $v\to \infty$, there exists a constant $C\in \R$ such that
    \[
    c(v)\geq Mv+C\quad \text{for all }v\geq 0.
    \]
    Hence, $c^*(M)\leq -C$, and it follows that
    \[
    c^{**}(v)\geq Mv-c^*(M)\geq Mv+C\quad\text{for all }v\geq 0.
    \]
    We have therefore shown that $\frac{c^{**}(v)}{v}\to \infty$ as $v\to \infty$.\qedhere
\end{enumerate}
\end{proof}

The next lemma plays a fundamental role in order to invoke the comparison result \cite[Theorem 4.9]{BlessingDenkKupperNendel2025}, which is used in Section \ref{sec.chernoff}.

\begin{lemma}\label{lem.uniqueness}
Let $(T_t)_{t\geq0}$ be a convex monotone semigroup in the sense of \cite[Definition 3.1]{BlessingDenkKupperNendel2025} with generator $B\colon D(B)\to \Cb(\R^d)$ such that $\Ccinf(\R^d)\subset D(B)$ and
\[
 Bf=\langle \nabla f,F\rangle +B_0f\quad \text{for all }f\in \Ccinf(\R^d),
\]
where $F\colon \R^d\to \R^d$ is Lipschitz and $B_0\colon \Ccinf(\R^d)\to \Cb(\R^d)$ satisfies
\[
\| B_0f\|_\infty\leq C\big(\|f\|_\infty+\|\nabla f\|_\infty+\|\nabla^2 f\|_\infty\big)\quad \text{for all }f\in \Ccinf(\R^d)
\]
with a constant $C\geq 0$ independent of $f\in \Ccinf(\R^d)$.\
Then, there exists a sequence $(\chi_n)_{n\in \N}\subset \Ccinf(\R^d)$ with $0\leq \chi_n\leq 1$ and $\chi_n(x)=1$ for all $n\in \N$ and $x\in \R^d$ with $|x|\leq n$ such that $\big(B(f \chi_n)\big)_{n\in \N}$ is uniformly bounded from above for all $f\in \Cbinf(\R^d)$ with $f\geq0$ and
\[
 \limsup_{h\downarrow 0}\sup_{x\in \R^d}\frac{(T_h f)(x)-f(x)}{h}<\infty.
\]
\end{lemma}

\begin{proof}
By \cite[Lemma A.1]{nendel2025chernoff} with $\kappa(s)=1+s$ for all $s\geq 0$, there exists a sequence $(\chi_n)_{n\in \N}\subset \Ccinf(\R^d)$ with $0\leq \chi_n\leq 1$, $\chi_n(x)=1$ for all $x\in \R^d$ with $|x|\leq n$, and 
 \begin{equation}\label{eq.bound.chin}
 \sup_{x\in \R^d}\big((1+|x|)\big|\nabla \chi_n(x)\big|+\big|\nabla^2 \chi_n(x)\big|\big)\leq 1
 \end{equation}
for all $n\in \N$. Now, let $f\in \Cbinf(\R^d)$ with $f\geq 0$ and 
 \[
  c_f:=\limsup_{h\downarrow 0}\sup_{x\in \R^d}\frac{(T_hf)(x)-f(x)}h<\infty.
 \]
 Then, $f_n:=f\chi_n\in \Ccinf(\R^d)$ for all $n\in \N$ and, by \cite[Lemma 4.8]{BlessingDenkKupperNendel2025}, 
 \[
 \sup_{|x|\leq n}(Bf_n)(x)\leq c_f\quad \text{for all }n\in \N.
\]
Moreover, \eqref{eq.bound.chin} together with the product rule implies that
\[
c:=\sup_{n\in \N}\big\|f\langle \nabla \chi_n,F\rangle+ B_0f_n\big\|_\infty<\infty,
\]
so that, again by the product rule,
\begin{align*}
 (Bf_n)(x)&=\big\langle \nabla f(x),F(x)\big\rangle+f(x)\big\langle \nabla \chi_n(x),F(x)\big\rangle+ (B_0f_n)(x)\\
 &\geq \langle \nabla f(x),F(x)\rangle-c
\end{align*}
for all $n\in \N$ and $x\in \R^d$ with $|x|\leq n$.\
Hence,
\[
\sup_{x\in \R^d}\big\langle \nabla f(x),F(x)\big\rangle\leq c+c_f,
\]
which implies that
\[
(Bf_n)(x)=\chi_n(x)\big\langle \nabla f(x),F(x)\big\rangle+f(x)\big\langle \nabla \chi_n(x),F(x)\big\rangle+ (B_0f_n)\leq 2c+c_f
\]
for all $x\in \R^d$ and $n\in \N$.
\end{proof}

We conclude with an auxiliary result that allows to verify a uniform equicontinuity condition  w.r.t.\ the mixed topology, which is needed in the context of the Chernoff approximation in Section \ref{sec.chernoff} and the definition of a convex monotone semigroup in \cite{BlessingDenkKupperNendel2025}.

\begin{lemma}\label{lem.cont.above}
 Let $h_0>0$ and $(I_h)_{h\in (0,h_0)}$ be a family of convex monotone operators such that, for all $\ep>0$, there exists a constant $C_\ep\geq 0$ with
\begin{equation}\label{eq.estimate.comapctness}
 \sup_{h\in (0,h_0)}\bigg\|\frac{I_hf-f}{h}\bigg\|_\infty\leq  \ep\|f\|_\infty+C_\ep\sup_{x\in \R^d}\big((1+|x|)\big|\nabla f(x)\big|+\big|\nabla^2f(x)\big|\big)
\end{equation}
for all $f\in \Ccinf(\R^d)$.\ 
Then, for all $\ep>0$, $r,T \geq0$ and $K\Subset \R^d$, there exist $c\geq 0$ and $K'\Subset\R^d$ with
$$\|I_{h}^kf-I_{h}^kg\|_{\infty,K} \leq c\|f-g\|_{\infty,K'}+\ep
$$
for all $h\in (0,h_0)$ and $k\in \N$ with $kh \leq T$ and all $f,g\in \Cb(\R^d)$ with $\|f\|_\infty\leq r$ and $\|g\|_\infty\leq r$.
\end{lemma}

\begin{proof}
Let $\ep>0$, $r\geq0$ and $K\Subset \R^d$.\ 
 By \cite[Lemma A.1]{nendel2025chernoff} with $\kappa(s)=1+s$ for all $s\geq 0$, there exists a function $\chi\in \Ccinf(\R^d)$ with $0\leq \chi\leq 1$, $\chi(x)=1$ for all $x\in K$ and 
 \begin{equation}\label{eq.bound.chi}
  \sup_{x\in \R^d}C_\ep\big((1+|x|)\big|\nabla \chi(x)\big|+\big|\nabla^2 \chi(x)\big|\big)\leq \ep.
 \end{equation}
 Then, by assumption, for all $h\in (0,h_0)$,
 \[
 \frac{\|I_h(r\chi)-(r\chi)\|_\infty}{h} \leq 2 r\ep.
 \]
 Choosing $K'\Subset \R^d$ as the support of $\chi$, the assumptions of \cite[Lemma 2.16]{BlessingKupperNendel2025} are satisfied with $\zeta_x:=\chi$ for all $x\in K$. The statement thus follows from \cite[Lemma 2.16]{BlessingKupperNendel2025}. 
\end{proof}

\end{document}